%% file: main.tex
\documentclass[11pt]{article}

\usepackage{etex}
\usepackage{enumitem}
\setlist{parsep = -0em, itemsep = 0.25em}
\usepackage[utf8]{inputenc}
\usepackage[dvips,letterpaper,margin=1in]{geometry}
\usepackage{amsmath}
\usepackage{amssymb}
\usepackage{amsthm}
\usepackage{bm}
\usepackage{colonequals}
\usepackage{comment}
\usepackage{graphicx}
\usepackage{subcaption}
\usepackage{xcolor}
\usepackage{tikz} 
\usetikzlibrary{shadows} 
\usepackage{authblk}

\usepackage{hyperref} 
\hypersetup{colorlinks}

\usepackage[
backend=biber,
style=alphabetic,
citestyle=alphabetic,
giveninits=true,
maxnames=99,
date=year
]{biblatex}
\DeclareFieldFormat[inproceedings]{title}{#1} 
\DeclareFieldFormat[article]{title}{#1}
\DeclareFieldFormat[inbook]{title}{#1}
\DeclareFieldFormat[misc]{title}{#1}
\DeclareFieldFormat[article]{year}{#1}
\DeclareFieldFormat[book]{title}{#1}

\renewbibmacro{in:}{}
\renewbibmacro*{volume+number+eid}{%
  \printfield{volume}
  \setunit{\addcolon\space}%
  \printfield{number}%
  \printfield{eid}}

\AtEveryBibitem{
    \clearfield{url}
    \clearfield{urldate}
    \clearfield{review}
    \clearfield{location}
    \clearfield{doi}
    \clearfield{series}
    \clearfield{publisher}
    \clearfield{issn}
    \clearfield{isbn}
}

\newtheorem{theorem}{Theorem}[section]
\newtheorem{lemma}[theorem]{Lemma}
\newtheorem{claim}[theorem]{Claim}

\newtheorem{proposition}[theorem]{Proposition}
\newtheorem{corollary}[theorem]{Corollary}

\newtheorem{definition}[theorem]{Definition}

\newcommand{\cA}{\mathcal{A}}
\newcommand{\cB}{\mathcal{B}}

\newcommand{\cD}{\mathcal{D}}

\newcommand{\cG}{\mathcal{G}}

\newcommand{\cN}{\mathcal{N}}
\newcommand{\cO}{\mathcal{O}}

\newcommand{\cS}{\mathcal{S}}

\newcommand{\cU}{\mathcal{U}}

\newcommand{\bI}{\bm I}

\newcommand{\bR}{\bm R}

\newcommand{\bW}{\bm W}

\newcommand{\ba}{\bm a}
\newcommand{\bb}{\bm b}
\newcommand{\bc}{\bm c}

\newcommand{\be}{\bm e}

\newcommand{\bh}{\bm h}

\newcommand{\bs}{\bm s}
\newcommand{\bt}{\bm t}

\newcommand{\bu}{\bm u}
\newcommand{\bv}{\bm v}
\newcommand{\bw}{\bm w}
\newcommand{\bx}{\bm x}
\newcommand{\by}{\bm y}
\newcommand{\bz}{\bm z}

\newcommand{\bzero}{\bm 0}
\newcommand{\bdd}{\mathrm{BDD}}
\newcommand{\clwe}{\mathrm{CLWE}}
\newcommand{\hclwe}{\mathrm{hCLWE}}
\newcommand{\dgs}{\mathrm{DGS}}

\newcommand{\E}{\operatorname*{\mathbb{E}}}
\newcommand{\poly}{\operatorname{poly}}
\newcommand{\linspan}{\operatorname{span}}

\newcommand{\eps}{\varepsilon}
\newcommand{\ngauss}[1]{D_{\mathbb{R}^n,#1}}

\newcommand{\what}{\widehat}

\title{Continuous LWE}
\author[a,b,c]{Joan Bruna\thanks{This work is partially supported by the Alfred P. Sloan Foundation, NSF RI-1816753, NSF CAREER CIF 1845360, and the Institute for Advanced Study.}}
\author[a]{Oded Regev\thanks{Research supported by the Simons Collaboration on Algorithms and Geometry, a Simons Investigator Award, and by the National Science Foundation (NSF) under Grant No.~CCF-1814524.}}%
\author[a]{Min Jae Song\thanks{Research supported by the National Science Foundation (NSF) under Grant No.~CCF-1814524.}}
\author[d]{Yi Tang\thanks{This work was done while the author was at the Courant Institute of Mathematical Sciences, New York University.}}
\affil[a]{Courant Institute of Mathematical Sciences, New York
  University, New York}
\affil[b]{Center for Data Science, New York University, New York}
\affil[c]{Institute for Advanced Study, Princeton}
\affil[d]{Computer Science and Engineering, University of Michigan, Ann Arbor}

\date{\today}

\begin{document}

\maketitle 

\begin{abstract}
We introduce a continuous analogue of the Learning with Errors (LWE) problem, which we name CLWE. 
We give a polynomial-time quantum reduction from worst-case lattice problems to CLWE,
showing that CLWE enjoys similar hardness guarantees to those of LWE.
Alternatively, our result can also be seen as opening new avenues of (quantum) attacks on lattice problems.
Our work resolves an open problem regarding the computational complexity of learning mixtures of Gaussians without separability assumptions (Diakonikolas 2016, Moitra 2018).
As an additional motivation, (a slight variant of) CLWE was considered in the context of robust machine learning (Diakonikolas et al.~FOCS 2017), where hardness in the statistical query (SQ) model was shown; our work addresses the open question regarding its computational hardness (Bubeck et al.~ICML 2019).
\end{abstract}

\section{Introduction}
\label{section:intro}
The Learning with Errors (LWE) problem has served as a foundation for many lattice-based cryptographic schemes~\cite{peikert2015decade}. Informally, LWE asks one to solve noisy random linear equations. To be more precise, 
the goal is to find a secret vector $\bs \in \mathbb{Z}_q^n$
given polynomially many samples of the form $(\ba_i, b_i)$, where $\ba_i \in \mathbb{Z}_q^n$ is uniformly chosen and $b_i \approx \langle \ba_i, \bs \rangle \pmod{q}$. In the absence of noise, LWE can be efficiently solved using Gaussian elimination. However, LWE is known to be hard assuming hardness of worst-case lattice problems such as Gap Shortest Vector Problem (GapSVP) or Shortest Independent Vectors Problem (SIVP) in the sense that there is a polynomial-time quantum reduction from these worst-case lattice problems to LWE~\cite{regev2005lwe}.

In this work, we introduce a new problem, called Continuous LWE (CLWE). As the name suggests, this problem can be seen as a continuous analogue of LWE, where equations in $\mathbb{Z}_q^n$ are replaced with vectors in $\mathbb{R}^n$ (see Figure~\ref{fig:plotinhom}). 
More precisely, CLWE considers noisy inner products $z_i \approx \gamma \langle \by_i, \bw \rangle \pmod{1}$, where the noise is drawn from a Gaussian distribution of width $\beta > 0$, $\gamma > 0$ is a problem parameter, $\bw \in \mathbb{R}^{n}$ is a secret unit vector, and the public vectors $\by_i \in \mathbb{R}^n$ are drawn from the standard Gaussian. Given polynomially many samples of the form $(\by_i, z_i)$, CLWE asks one to find the secret direction $\bw$.

\begin{figure}[ht]
\centering
\includegraphics[width=0.6\textwidth]{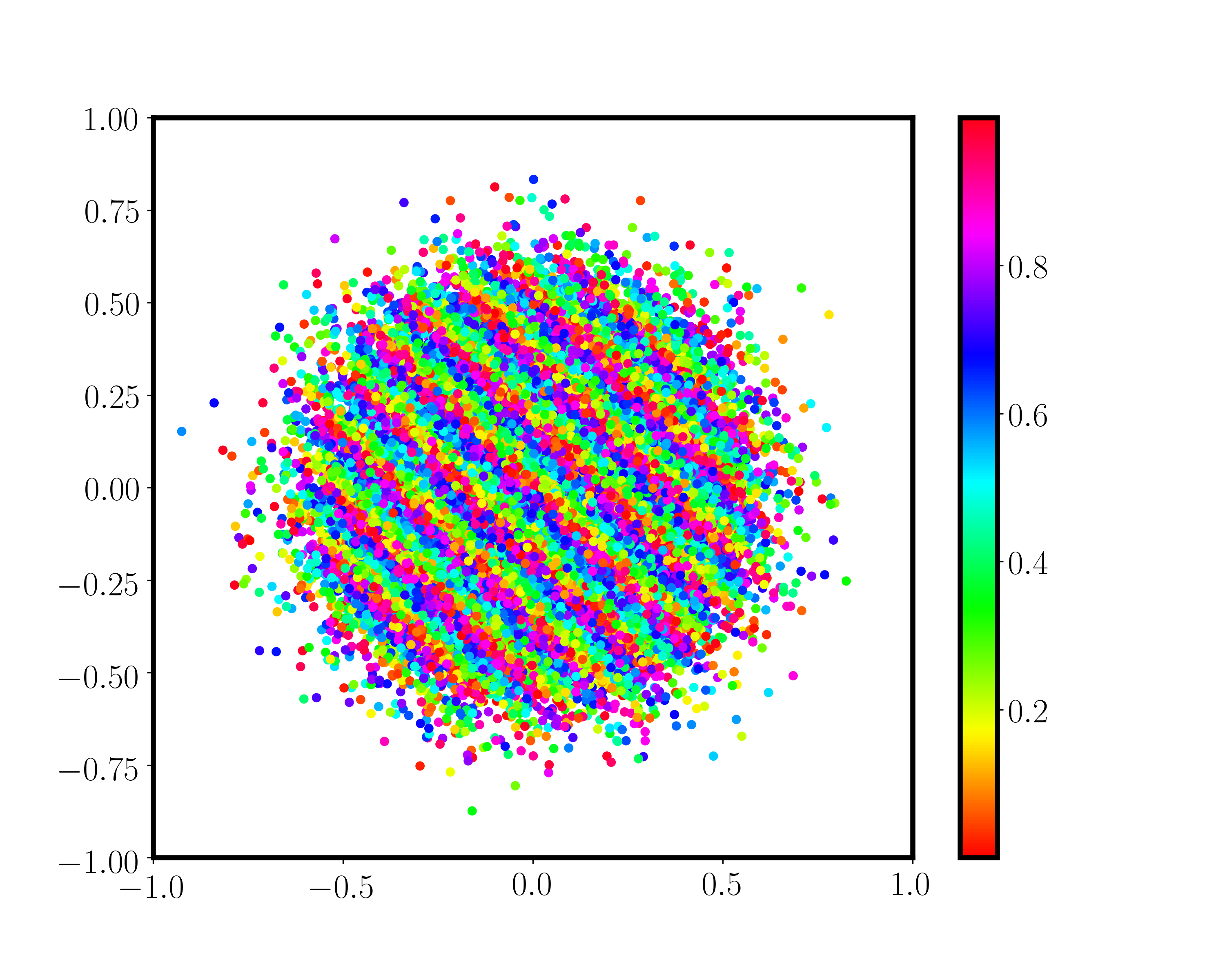}
\caption{Scatter plot of two-dimensional CLWE samples. Color indicates the last ($z$) coordinate.}
\label{fig:plotinhom}
\end{figure}

One can also consider a closely related homogeneous variant of CLWE (see Figure~\ref{fig:plothom}). This distribution, which we call homogeneous CLWE, can be obtained by essentially conditioning on $z_i \approx 0$. It is a mixture of ``Gaussian pancakes'' of width $\approx \beta/\gamma$ in the secret direction and width $1$ in the 
remaining $n-1$ directions. The Gaussian components are equally spaced, with a separation of $\approx 1/\gamma$. (See Definition~\ref{def:hclwe} for the precise statement.)

\begin{figure}[ht]
\centering
\begin{minipage}[t]{0.36\textwidth}
\includegraphics[width=\textwidth]{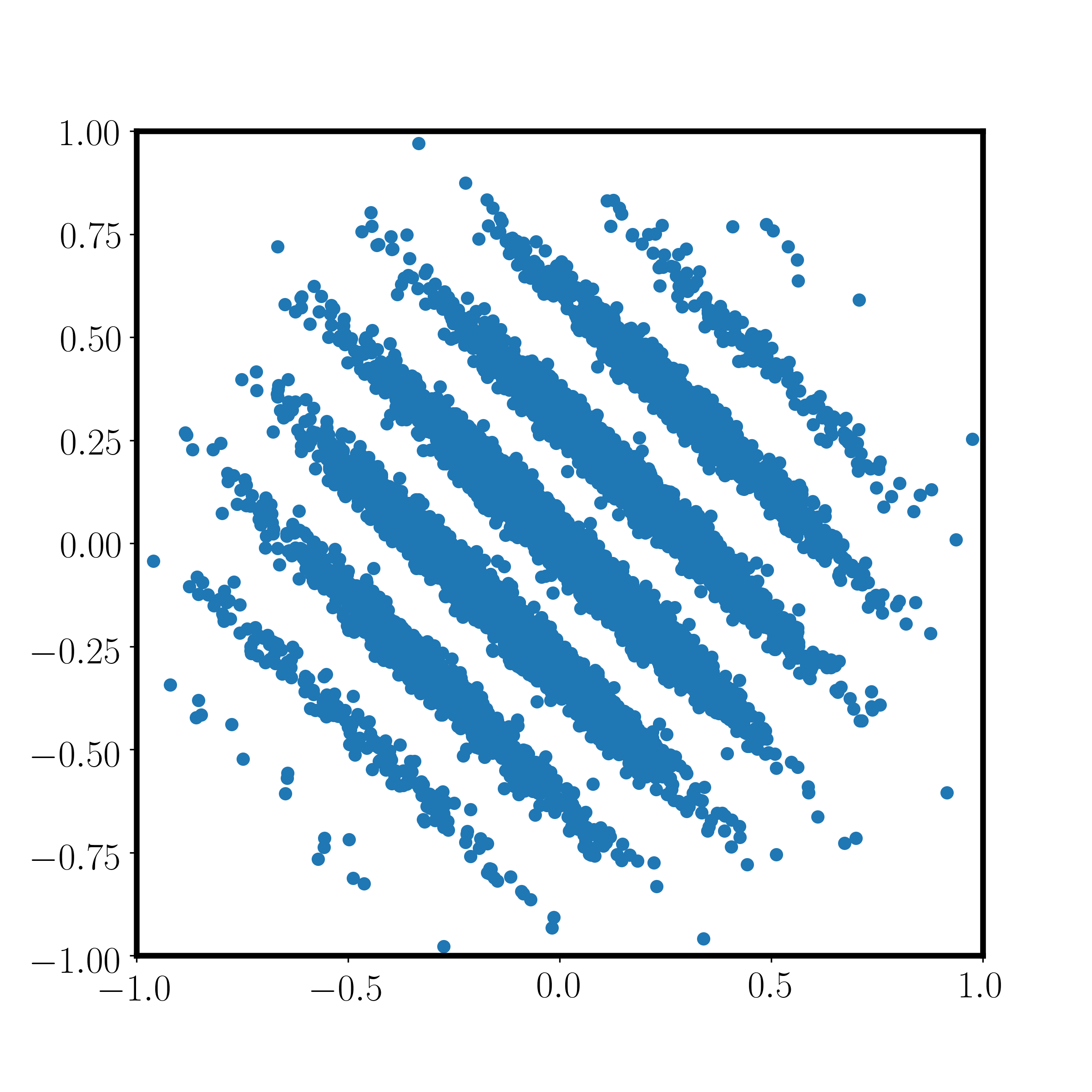}
\end{minipage}
\begin{minipage}[t]{0.54\textwidth}
\includegraphics[width=\textwidth]{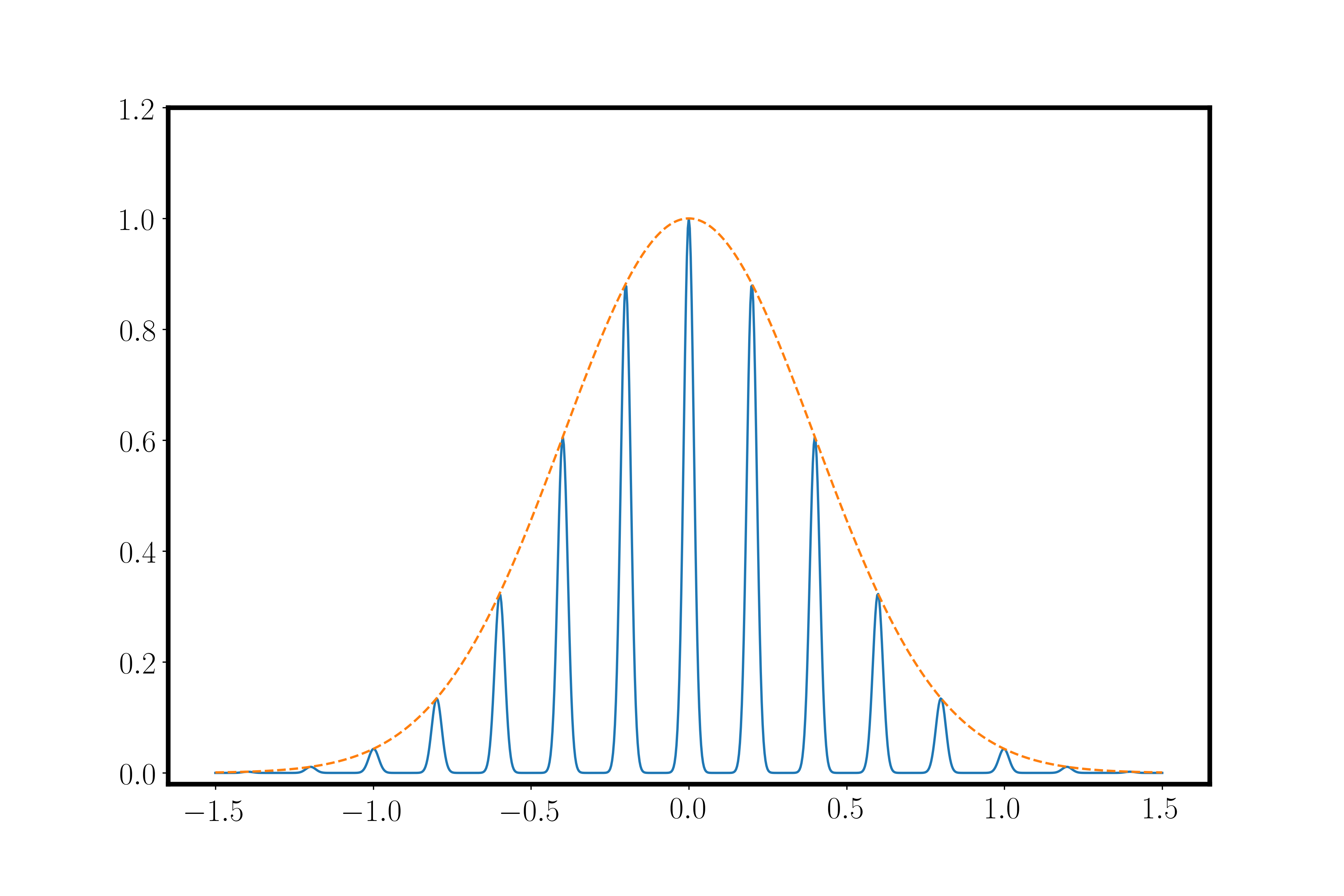}
\end{minipage}\hspace{0.05\textwidth}
\caption{Left: Scatter plot of two-dimensional homogeneous CLWE samples.
Right: Unnormalized probability densities of homogeneous CLWE (blue) and Gaussian (orange) along the hidden direction.}
\label{fig:plothom}
\end{figure}

Our main result is that CLWE (and homogeneous CLWE) enjoy hardness guarantees similar to those of LWE.

\begin{theorem}[Informal]
\label{thm:main-informal}
Let $n$ be an integer, $\beta = \beta(n) \in (0,1)$ and $\gamma = \gamma(n) \geq 2\sqrt{n}$ such that the ratio $\gamma/\beta$ is polynomially bounded. 
If there exists an efficient algorithm that solves $\clwe_{\beta, \gamma}$, then there exists an efficient quantum algorithm that approximates worst-case lattice problems to within polynomial factors.
\end{theorem}

Although we defined CLWE above as a search problem of finding the hidden direction, 
Theorem~\ref{thm:main-informal} is actually stronger, and applies to the decision variant of CLWE in which the goal is to distinguish CLWE samples $(\by_i, z_i)$ from samples where the noisy inner product $z_i$ is replaced by a random number distributed uniformly on $[0,1)$ (and similarly for the homogeneous variant). 

\paragraph{Motivation: Lattice algorithms.}
Our original motivation to consider CLWE is as a possible approach to finding quantum algorithms for lattice problems. Indeed, the reduction above (just like the reduction to LWE~\cite{regev2005lwe}), can be interpreted in an algorithmic way: in order to quantumly solve worst-case lattice problems, ``all'' we have to do is solve CLWE (classically or quantumly). The elegant geometric nature of CLWE opens up a new toolbox of techniques that can potentially be used for solving lattice problems,
such as sum-of-squares-based techniques and algorithms for learning mixtures of Gaussians~\cite{moitrav2010mixture}.
Indeed, some recent algorithms (e.g.,~\cite{klivanskothari2019list-dec,raghavendrayau2020list-dec}) solve problems that include CLWE  or homogeneous CLWE as a special case (or nearly so), yet as far as we can tell, so far none of the known results leads to an improvement over the state of the art in lattice algorithms.

To demonstrate the usefulness of CLWE as an algorithmic target, we show in Section~\ref{section:subexp} a simple moment-based algorithm that solves CLWE in time $\exp(\gamma^2)$.
Even though this does not imply subexponential time algorithms for lattice problems (since Theorem~\ref{thm:main-informal} requires $\gamma > \sqrt{n}$), it is interesting to contrast this algorithm with an analogous algorithm for LWE by Arora and Ge~\cite{arora2011subexplwe}. The two algorithms have the same running time (where $\gamma$ is replaced by the absolute noise $\alpha q$ in the LWE samples), and both rely on related techniques (moments in our case, powering in Arora-Ge's), yet the Arora-Ge algorithm is technically more involved than our rather trivial algorithm (which just amounts to computing the empirical covariance matrix). We interpret this as an encouraging sign that CLWE might be a better algorithmic target than LWE. 

\paragraph{Motivation: Hardness of learning Gaussian mixtures.}
Learning mixtures of Gaussians is a classical problem in machine learning~\cite{pearson1984gmm}. Efficient algorithms are known for the task if the Gaussian components are guaranteed to be sufficiently well separated (e.g.,~\cite{dasgupta1999gmm,vempala-wang2002spectralgmm,arora-kannan2005,dasgupta-schulman2007em,brubaker-vempala2008pca,regev2017gmm,hopkins2018gmm,kothari-steinhardt2018clustering,diakonikolas2018spherical-gmm}).
Without such strong separation requirements, it is known that efficiently recovering the individual components of a mixture (technically known as ``parameter estimation") is in general impossible~\cite{moitrav2010mixture}; intuitively, this exponential information theoretic lower bound holds because the Gaussian components ``blur into each other", despite being mildly separated pairwise. 

This leads to the question of whether there exists an efficient algorithm that can learn mixtures of Gaussians without strong separation requirement, not in the above strong parameter estimation sense (which is impossible), but rather in the much weaker density estimation sense, where the goal is merely to output an approximation of the given distribution's density function. See~\cite{diakonikolas2016structured,moitra2018} for the precise statement and~\cite{diakonikolas2017sqgaussian} where a super-polynomial lower bound for density estimation is shown in the restricted statistical query (SQ) model~\cite{kearnsSQ1998,feldman2017planted-clique}. Our work provides a negative answer to this open question,
showing that learning Gaussian mixtures is computationally difficult even if the goal is only to output an estimate of the density (see Proposition~\ref{prop:mixture-learning-hardness}). It is worth noting that our hard instance has almost non-overlapping components, i.e., the pairwise statistical distance between distinct Gaussian components is essentially 1, a property shared by the SQ-hard instance of~\cite{diakonikolas2017sqgaussian}. 

\paragraph{Motivation: Robust machine learning.}
Variants of CLWE have already been analyzed in the context of robust machine learning~\cite{bubeck2019}, in which the goal is to learn a classifier that is robust against adversarial examples at test time~\cite{szegedy2014adversarial-examples}. In particular, Bubeck et al.~\cite{bubeck2019} use the SQ-hard Gaussian mixture instance of Diakonikolas et al.~\cite{diakonikolas2017sqgaussian} to establish SQ lower bounds for 
learning a certain binary classification task, which can be seen as a variant of homogeneous CLWE. The key difference between our distribution and that of~\cite{diakonikolas2017sqgaussian,bubeck2019} is that our distribution has equal spacing between the ``layers" along the hidden direction, whereas their ``layers" are centered around roots of Hermite polynomials (the goal being to exactly match the lower moments of the standard Gaussian). The connection to lattices, which we make for the first time here, answers an open question by Bubeck et al.~\cite{bubeck2019}. 

As additional evidence of the similarity between homogeneous CLWE and the distribution considered in~\cite{diakonikolas2017sqgaussian, bubeck2019}, we prove a super-polynomial SQ lower bound for homogeneous CLWE (even with super-polynomial precision). For $\gamma=\Omega(\sqrt{n})$, this result translates to an exponential SQ lower bound for exponential precision, which corroborates our computational hardness result based on worst-case lattice problems. The uniform spacing in the hidden structure of homogeneous CLWE leads to a simplified proof of the SQ lower bound compared to previous works, which considered non-uniform spacing between the Gaussian components. Note that computational hardness does not automatically imply SQ hardness as query functions in the SQ framework need not be efficiently computable.

Bubeck et al.~\cite{bubeck2019} were also interested in a variant of the learning problem where instead of \emph{one} hidden direction, there are $m \ge 1$ orthogonal hidden directions. So, for instance, the ``Gaussian pancakes'' in the $m=1$ case above are replaced with ``Gaussian baguettes'' in the case $m=2$, forming an orthogonal grid in the secret two-dimensional space. As we show in Section~\ref{section:k-hc}, our computational hardness easily extends to the $m>1$ case using a relatively standard hybrid argument. The same is true for the SQ lower bound we show in Section~\ref{section:sq-lb} (as well as for the SQ lower bound in~\cite{diakonikolas2017sqgaussian,bubeck2019}; the proof is nearly identical). The advantage of the $m>1$ variant is that the distance between the Gaussian mixture components increases from $\approx 1/\gamma$ (which can be as high as $\approx 1/\sqrt{n}$ if we want our hardness to hold) to $\approx \sqrt{m}/\gamma$ (which can be as high as $\approx 1$ by taking $m \approx n$). This is a desirable feature for showing hardness of robust machine learning. 

\paragraph{Motivation: Cryptographic applications.}
Given the wide range of cryptographic applications of LWE~\cite{peikert2015decade}, it is only natural to expect that CLWE would also be useful for some cryptographic tasks, a question we leave for future work. CLWE's clean and highly symmetric definition should make it a better fit for some applications; its continuous nature, however, might require a discretization step due to efficiency considerations. 

\paragraph{Analogy with LWE.}
As argued above, there are apparently nontrivial differences between CLWE and LWE, especially in terms of possible algorithmic approaches. However, there is undoubtedly also strong similarity between the two. 
In terms of parameters, the $\gamma$ parameter in CLWE (density of layers) plays the role of the absolute noise level $\alpha q$ in LWE. And the $\beta$ parameter in CLWE plays the role of the relative noise parameter $\alpha$ in LWE. Using this correspondence between the parameters, the hardness proved for CLWE in Theorem~\ref{thm:main-informal} is essentially identical to the one proved for LWE in~\cite{regev2005lwe}. The similarity extends even to the noiseless case, where $\alpha = 0$ in LWE and $\beta = 0$ in CLWE. In particular, in Section~\ref{section:lll-clwe} we present an efficient LLL-based algorithm for solving noiseless CLWE, which is analogous to Gaussian elimination for noiseless LWE.

\paragraph{Comparison with previous work.}
The CLWE problem is related to the hard problem introduced in the seminal work of Ajtai and Dwork~\cite{ajtai97adcrypto}. Specifically, both problems involve finding a hidden direction in samples from a continuous distribution. One crucial difference, though, is in the density of the layers. Whereas in our hardness result the separation between the layers can be as large as $\approx 1/\sqrt{n}$, in Ajtai and Dwork the separation is exponentially small. This larger separation in CLWE is more than just a technicality. First, it is the reason we need to employ the quantum machinery from the LWE hardness proof~\cite{regev2005lwe}. Second, it is nearly tight, as demonstrated by the algorithm in Section~\ref{section:subexp}. Third, it is necessary for applications such as hardness of learning Gaussian mixtures. Finally, this larger separation is analogous to the main difference between LWE and earlier work~\cite{regev2004harmonic}, and is what leads to the relative efficiency of LWE-based cryptography. 

\paragraph{Acknowledgements.} 
We thank Aravindan Vijayaraghavan and Ilias Diakonikolas for useful comments.

\subsection{Technical Overview}
\label{section:technical-overview}
Broadly speaking, our proof follows the iterative structure of the original LWE hardness proof~\cite{regev2005lwe} (in fact, one might say most of the ingredients for CLWE were already present in that 2005 paper!). 
We also make use of some recent techniques, such as a way to reduce to decision problems directly~\cite{peikert2017ringlwe}.

In more detail, as in previous work, 
our main theorem boils down to solving the following problem: we are given a $\clwe_{\beta,\gamma}$ oracle and polynomially many samples from $D_{L,r}$, the 
discrete Gaussian distribution on $L$ of width $r$,%
\footnote{We actually require samples from $D_{L,r_i}$ for polynomially many $r_i$'s satisfying $r_i \geq r$, see Section~\ref{section:clwe-hardness}.} and our goal is to solve $\bdd_{L^*,\gamma/r}$, which is the problem of finding the closest vector in the dual lattice $L^*$ given a vector $\bt$ that is within distance $\gamma/r$ of $L^*$. (It is known that $\bdd_{L^*,1/r}$ can be efficiently solved even if all we are given is polynomially many samples from $D_{L,r}$, without any need for an oracle~\cite{aharonov2005conp}; the point here is that the CLWE oracle allows us to extend the decoding radius from $1/r$ to $\gamma/r$.)
Once this is established, the main theorem follows from previous work~\cite{peikert2017ringlwe,regev2005lwe}. Very briefly, the resulting BDD solution is used in a quantum procedure to produce discrete Gaussian samples that are shorter than the ones we started with. This process is then repeated, until eventually we end up with the desired short discrete Gaussian samples. We remark that this process incurs a $\sqrt{n}$ loss in the Gaussian width (Lemma~\ref{lem:reg05quantumstep}), and the reason we require $\gamma \ge 2\sqrt{n}$ is to overcome this loss. 

We now explain how we solve the above problem. For simplicity, assume for now that we have a \emph{search} CLWE oracle that recovers the secret exactly. (Our actual reduction is stronger and only requires a \emph{decision} CLWE oracle.) Let the given BDD instance be $\bu + \bw$, where $\bu \in L^*$ and $\|\bw\| = \gamma/r$. We will consider the general case of $\|\bw\| \le \gamma/r$ in Section~\ref{section:clwe-hardness}.
The main idea is to generate CLWE samples whose secret is essentially the desired BDD solution $\bw$, which would then complete the proof. To begin, take a sample from the discrete Gaussian distribution $\by \sim D_{L,r}$ (as provided to us) and consider the inner product
\begin{align*}
    \langle \by, \bu + \bw \rangle = \langle \by, \bw \rangle \pmod 1 \; ,
\end{align*}
where the equality holds since $\langle \by, \bu \rangle \in \mathbb{Z}$ by definition.
The $(n+1)$-dimensional vector $(\by, \langle \by, \bw \rangle \bmod 1)$ is almost a CLWE sample (with parameter $\gamma$ since $\gamma = r\|\bw\|$ is the width of $\langle \by, \bw \rangle$) --- the only problem is that in CLWE the $\by$'s need to be distributed according to a standard Gaussian, but here the $\by$'s are distributed according to a \emph{discrete} Gaussian over $L$. To complete the transformation into bonafide CLWE samples, we add Gaussian noise of appropriate variance to both $\by$ and $\langle \by, \bw \rangle$ (and rescale $\by$ so that it is distributed according to the standard Gaussian distribution). We then apply the search $\clwe_{\beta,\gamma}$ oracle on these CLWE samples to recover $\bw$ and thereby solve $\bdd_{L^*,\gamma/r}$.

As mentioned previously, our main result actually uses a \emph{decision} CLWE oracle, which does not recover the secret $\bw$ immediately. Working with this decision oracle requires some care. To that end, our proof will incorporate the ``oracle hidden center'' finding procedure from~\cite{peikert2017ringlwe}, the details of which can be found in Section~\ref{section:solve-bdd-with-clwe}.

\section{Preliminaries}

\begin{definition}[Statistical distance] For two distributions $\cD_1$ and $\cD_2$ over $\mathbb{R}^n$ with density functions $\phi_1$ and $\phi_2$, respectively, we define the \emph{statistical distance} between them as
\begin{align*}
    \Delta(\cD_1,\cD_2) = \frac{1}{2}\int_{\mathbb{R}^n}|\phi_1(\bx)-\phi_2(\bx)|d\bx
    \; .
\end{align*}
\end{definition}
We denote the statistical distance by $\Delta(\phi_1,\phi_2)$ if only the density functions are specified.
Moreover, for random variables $X_1 \sim \cD_1$ and $X_2 \sim \cD_2$, we also denote $\Delta(X_1,X_2) = \Delta(\cD_1,\cD_2)$. One important fact is that applying (possibly a randomized) function cannot increase statistical distance, i.e., for random variables $X, Y$ and function $f$,
\begin{align*}
    \Delta(f(X),f(Y)) \leq \Delta(X,Y)
    \; .
\end{align*}

We define the \emph{advantage} of an algorithm $\cA$ solving the decision problem of distinguishing two distributions $\cD_n$ and $\cD'_n$ parameterized by $n$ as
\begin{align*}
    \Bigl| \Pr_{x \sim \cD_n}[\cA(x) = \mathrm{YES}] - \Pr_{x \sim \cD'_n}[\cA(x) = \mathrm{YES}] \Bigr|
    \; .
\end{align*}
Moreover, we define the \emph{advantage} of an algorithm $\cA$ solving the \emph{average-case} decision problem of distinguishing two distributions $\cD_{n, s}$ and $\cD'_{n, s}$ parameterized by $n$ and $s$, where $s$ is equipped with some distribution $\cS_n$, as
\begin{align*}
    \Bigl| \Pr_{s \sim \cS_n}[\cA^{\cB_{n, s}}(1^n) = \mathrm{YES}] - \Pr_{s \sim \cS_n}[\cA^{\cB'_{n, s}}(1^n) = \mathrm{YES}] \Bigr|
    \; ,
\end{align*}
where $\cB_{n, s}$ and $\cB_{n, s}$ are respectively the sampling oracles of $\cD_{n, s}$ and $\cD'_{n, s}$.
We say that an algorithm $\cA$ has \emph{non-negligible advantage} if its advantage is a non-negligible function in $n$, i.e., a function in $\Omega(n^{-c})$ for some constant $c > 0$.

\subsection{Lattices and Gaussians}
\paragraph{Lattices.} 
A \emph{lattice} is a discrete additive subgroup of $\mathbb{R}^n$.
Unless specified otherwise, we assume all lattices are full rank, i.e., their linear span is $\mathbb{R}^n$.
For an $n$-dimensional lattice $L$, a set of linearly independent vectors $\{\bb_1, \dots, \bb_n\}$ is called a \emph{basis} of $L$ if $L$ is generated by the set, i.e., $L = B \mathbb{Z}^n$ where $B = [\bb_1, \dots, \bb_n]$.
The \emph{determinant} of a lattice $L$ with basis $B$ is defined as $\det(L) = |\det(B)|$; it is easy to verify that the determinant does not depend on the choice of basis.

The \emph{dual lattice} of a lattice $L$, denoted by $L^*$, is defined as
\begin{align*}
    L^* = \{ \by \in \mathbb{R}^n \mid \langle \bx, \by \rangle \in \mathbb{Z} \text{ for all } \bx \in L\}
    \; .
\end{align*}
If $B$ is a basis of $L$ then $(B^T)^{-1}$ is a basis of $L^*$; in particular, $\det(L^*) = \det(L)^{-1}$.

\begin{definition} For an $n$-dimensional lattice $L$ and $1 \le i \le n$, the \emph{$i$-th successive minimum} of $L$ is defined as
\begin{align*}
    \lambda_i(L) = \inf \{r \mid \dim(\linspan(L \cap \overline{B}(\bzero,r))) \geq i\}
    \; ,
\end{align*}
where $\overline{B}(\bzero,r)$ is the closed ball of radius $r$ centered at the origin.
\end{definition}

We define the function $\rho_s(\bx) = \exp(-\pi\|\bx/s\|^2)$. Note that $\rho_s(\bx) / s^n$, where $n$ is the dimension of $\bx$, is the probability density of the Gaussian distribution with covariance $s^2/(2\pi)\cdot I_n$.
\begin{definition}[Discrete Gaussian] For lattice $L \subset \mathbb{R}^n$, vector $\by \in \mathbb{R}^n$, and parameter $r > 0$, the \emph{discrete Gaussian distribution} $D_{\by+L,r}$ on coset $\by+L$ with width $r$ is defined to have support $\by+L$ and probability mass function proportional to $\rho_r$.
\end{definition}
For $\by = \bm 0$, we simply denote the discrete Gaussian distribution on lattice $L$ with width $r$ by $D_{L,r}$.
Abusing notation, we denote the $n$-dimensional \emph{continuous Gaussian distribution} with zero mean and isotropic variance $r^2/(2\pi)$ as $D_{\mathbb{R}^n,r}$. 
Finally, we omit the subscript $r$ when $r = 1$ and refer to $D_{\mathbb{R}^n}$ as the \emph{standard} Gaussian (despite it having covariance $I_n/(2\pi)$).

\begin{claim}[{\cite[Fact 2.1]{peikert2010sampler}}]
\label{claim:complete-squares}
For any $r_1, r_2 > 0$ and vectors $\bx, \bc_1, \bc_2 \in \mathbb{R}^n$,
let $r_0 = \sqrt{r_1^2 + r_2^2}$, $r_3 = r_1 r_2 / r_0$, and $\bc_3 = (r_3/r_1)^2 \bc_1 + (r_3/r_2)^2 \bc_2$.
Then
\begin{align*}
    \rho_{r_1}(\bx-\bc_1) \cdot \rho_{r_2}(\bx - \bc_2) = \rho_{r_0}(\bc_1 - \bc_2) \cdot \rho_{r_3}(\bx-\bc_3)
    \; .
\end{align*}
\end{claim}

\paragraph{Fourier analysis.} We briefly review basic tools of Fourier analysis required later on. The Fourier transform of a function $f: \mathbb{R}^n \to \mathbb{C}$ is defined to be
\begin{align*}
    \hat{f}(\bw) = \int_{\mathbb{R}^n} f(\bx)e^{-2\pi i \langle \bx, \bw \rangle}d\bx
    \; .
\end{align*}

An elementary property of the Fourier transform is that if $f(\bw) = g(\bw+\bv)$ for some $\bv \in \mathbb{R}^n$, then $\hat{f}(\bw) = e^{2\pi i \langle \bv, \bw \rangle}\hat{g}(\bw)$. Another important fact is that the Fourier transform of a Gaussian is also a Gaussian, i.e., $\hat{\rho} = \rho$; more generally, $\hat{\rho}_s = s^n \rho_{1/s}$.
We also exploit the Poisson summation formula stated below. Note that we denote by $f(A) = \sum_{\bx \in A} f(\bx)$ for any function $f$ and any discrete set $A$. 

\begin{lemma}[Poisson summation formula] For any lattice $L$ and any function $f$,\footnote{To be precise, $f$ needs to satisfy some niceness conditions; this will always hold in our applications.}
\label{lem:poisson-sum}
\begin{align*}
    f(L) = \det(L^*)\cdot \what{f}(L^*)
    \; .
\end{align*}
\end{lemma}

\paragraph{Smoothing parameter.} An important lattice parameter induced by discrete Gaussian which will repeatedly appear in our work is the \emph{smoothing parameter}, defined as follows.
\begin{definition}[Smoothing parameter]
For lattice $L$ and real $\eps > 0$, we define the \emph{smoothing parameter} $\eta_{\eps}(L)$ as
\begin{align*}
    \eta_{\eps}(L) = \inf \{s \mid \rho_{1/s}(L^* \setminus \{\bzero\}) \leq \eps\}
    \; .
\end{align*}
\end{definition}

Intuitively, this parameter is the width beyond which the discrete Gaussian distribution behaves like a continuous Gaussian. This is formalized in the lemmas below.

\begin{lemma}[{\cite[Claim 3.9]{regev2005lwe}}]
\label{lem:smoothing-gaussian}
For any $n$-dimensional lattice $L$, vector $\bu \in \mathbb{R}^n$, and $r,s > 0$ satisfying $rs/t \geq \eta_\eps(L)$ for some $\eps < \frac{1}{2}$, where $t = \sqrt{r^2+s^2}$, the statistical distance between $D_{\bu+L,r} + D_{\mathbb{R}^n, s}$ and $D_{\mathbb{R}^n, t}$ is at most $4\eps$.
\end{lemma}

\begin{lemma}[{\cite[Lemma 2.5]{peikert2017ringlwe}}]
\label{lem:smoothing-uniform}
For any $n$-dimensional lattice $L$, real $\eps > 0$, and $r \geq \eta_{\eps}(L)$, the statistical distance between $D_{\mathbb{R}^n, r} \bmod L$ and the uniform distribution over $\mathbb{R}^n / L$ is at most $\eps/2$.
\end{lemma}

Lemma~\ref{lem:smoothing-gaussian} states that if we take a sample from $D_{L,r}$ and add continuous Gaussian noise $D_{\mathbb{R}^n,s}$ to the sample, the resulting distribution is statistically close to $D_{\mathbb{R}^n,\sqrt{r^2+s^2}}$, which is precisely what one gets by adding two continuous Gaussian distributions of width $r$ and $s$. Unless specified otherwise, we always assume $\eps$ is negligibly small in $n$, say $\eps = \exp(-n)$. The following are some useful upper and lower bounds on the smoothing parameter $\eta_\eps(L)$.

\begin{lemma}[{\cite[Lemma 2.6]{peikert2017ringlwe}}]
\label{lem:smoothing-dual}
For any $n$-dimensional lattice $L$ and $\eps = \exp(-c^2n)$,
\begin{align*}
    \eta_\eps(L) \leq c\sqrt{n}/\lambda_1(L^*)
    \; .
\end{align*}
\end{lemma}

\begin{lemma}[{\cite[Lemma 3.3]{micciancio2007average}}]
\label{lem:smoothing-primal}
For any $n$-dimensional lattice $L$ and $\eps > 0$,
\begin{align*}
    \eta_\eps(L) \leq \sqrt{\frac{\ln(2n(1+1/\eps))}{\pi}}\cdot \lambda_n(L)
    \; .
\end{align*}
\end{lemma}

\begin{lemma}[{\cite[Claim 2.13]{regev2005lwe}}]
\label{lem:smoothing-lb}
For any $n$-dimensional lattice $L$ and $\eps > 0$,
\begin{align*}
    \eta_\eps(L) \geq \sqrt{\frac{\ln 1/\eps}{\pi}}\cdot\frac{1}{\lambda_1(L^*)}
    \; .
\end{align*}
\end{lemma}

\paragraph{Computational problems.}
GapSVP and SIVP are among the main computational problems on lattices and are believed to be computationally hard (even with quantum computation) for polynomial approximation factor $\alpha(n)$. We also define two additional problems, DGS and BDD. 

\begin{definition}[GapSVP] 
For an approximation factor $\alpha = \alpha(n)$, an instance of $\mathrm{GapSVP}_\alpha$ is given by an $n$-dimensional lattice $L$ and a number $d > 0$. In \textnormal{YES} instances, $\lambda_1(L) \leq d$, whereas in \textnormal{NO} instances, $\lambda_1(L) > \alpha \cdot d$.
\end{definition}

\begin{definition}[SIVP]
For an approximation factor $\alpha = \alpha(n)$, an instance of $\mathrm{SIVP}_\alpha$ is given by an $n$-dimensional lattice $L$. The goal is to output a set of $n$ linearly independent lattice vectors of length at most $\alpha \cdot \lambda_n(L)$.
\end{definition}

\begin{definition}[DGS]
For a function $\varphi$ that maps lattices to non-negative reals, an instance of $\dgs_\varphi$ is given by a lattice $L$ and a parameter $r \geq \varphi(L)$. The goal is to output an independent sample whose distribution is within negligible statistical distance of $D_{L,r}$.
\end{definition}

\begin{definition}[BDD]
For an $n$-dimensional lattice $L$ and distance bound $d > 0$, an instance of $\bdd_{L,d}$  is given by a vector $\bt = \bw + \bu$, where $\bu \in L$ and $\|\bw\| \leq d$. The goal is to output $\bw$.
\end{definition}

\subsection{Learning with errors}
\label{prelim:lwe}
We now define the learning with errors (LWE) problem. This definition will not be used in the sequel, and is included for completeness. Let $n$ and $q$ be positive integers, and $\alpha > 0$ an error rate. We denote the quotient ring of integers modulo $q$ as $\mathbb{Z}_q = \mathbb{Z}/q\mathbb{Z}$ and quotient group of reals modulo the integers as $\mathbb{T} = \mathbb{R}/\mathbb{Z} = [0, 1)$.

\begin{definition}[LWE distribution] For integer $q \ge 2$ and vector $\bs \in \mathbb{Z}_q^n$, the \emph{LWE distribution} $A_{\bs, \alpha}$ over $\mathbb{Z}_q^n \times \mathbb{T}$ is sampled by independently choosing uniformly random $\ba \in \mathbb{Z}_q^n$ and $e \sim D_{\mathbb{R},\alpha}$, and outputting $(\ba, (\langle \ba, \bs \rangle/q + e) \bmod 1)$.
\end{definition}

\begin{definition} For an integer $q = q(n) \geq 2$ and error parameter $\alpha = \alpha(n) > 0$, the average-case decision problem $\mathrm{LWE}_{q,\alpha}$ is to distinguish the following two distributions over $\mathbb{Z}_q^n \times \mathbb{T}$: (1) the LWE distribution $A_{\bs, \alpha}$ for some uniformly random $\bs \in \mathbb{Z}_q^n$ (which is fixed for all samples), or (2) the uniform distribution.
\end{definition}

\subsection{Continuous learning with errors}
\label{prelim:comb}
We now define the CLWE distribution, which is the central subject of our analysis.

\begin{definition}[CLWE distribution]
For unit vector $\bw \in \mathbb{R}^n$ and parameters $\beta, \gamma > 0$, define the \emph{CLWE distribution} $A_{{\bw}, \beta, \gamma}$ over $\mathbb{R}^{n+1}$ to have density at $(\by,z)$ proportional to
\begin{align*}
    \rho(\by) \cdot \sum_{k \in \mathbb{Z}} \rho_\beta(z+k-\gamma\langle \by, \bw \rangle)
    \; .
\end{align*}
\end{definition}

Equivalently, a sample $(\by, z)$ from the CLWE distribution $A_{\bw, \beta, \gamma}$ is given by the $(n+1)$-dimensional vector $(\by, z)$ where $\by \sim D_{\mathbb{R}^n}$ and 
    $z = (\gamma \langle \by, \bw \rangle + e) \bmod 1$ where $e \sim D_{\mathbb{R},\beta}$.
The vector $\bw$ is the hidden direction, $\gamma$ is the density of layers, and $\beta$ is the noise added to each equation. From the CLWE distribution, we can arrive at the homogeneous CLWE distribution by conditioning on $z = 0$. A formal definition is given as follows.

\begin{definition}[Homogeneous CLWE distribution]\label{def:hclwe}
For unit vector $\bw \in \mathbb{R}^n$ and parameters $\beta, \gamma > 0$, define the \emph{homogeneous CLWE distribution} $H_{\bw, \beta, \gamma}$ over $\mathbb{R}^n$ to have density at $\by$ proportional to
\begin{align}\label{eqn:hclwe-def}
    \rho(\by) \cdot \sum_{k \in \mathbb{Z}} \rho_\beta(k-\gamma\langle \by, \bw \rangle)
    \; .
\end{align}
\end{definition}

The homogeneous CLWE distribution can be equivalently defined as a mixture of Gaussians.
To see this, notice that Eq.~\eqref{eqn:hclwe-def} is equal to
\begin{align}\label{eqn:hclwe-mixture-def}
    \sum_{k \in \mathbb{Z}} \rho_{\sqrt{\beta^2+\gamma^2}}(k) \cdot 
    \rho(\pi_{\bw^\perp}(\by)) \cdot \rho_{\beta/\sqrt{\beta^2+\gamma^2}}\Big(\langle \by, \bw \rangle -\frac{\gamma}{\beta^2+\gamma^2}k\Big) \; ,
\end{align}
where $\pi_{\bw^\perp}$ denotes the projection on the orthogonal space to $\bw$.
Hence, $H_{\bw, \beta, \gamma}$ can be viewed as a mixture of Gaussian components of width 
$\beta/\sqrt{\beta^2+\gamma^2}$ (which is roughly $\beta/\gamma$ for $\beta \ll \gamma$) in the secret direction, and width $1$ in the orthogonal space. The components are equally spaced, with a separation of $\gamma/(\beta^2+\gamma^2)$ between them (which is roughly $1/\gamma$ for $\beta \ll \gamma$). 

We remark that the integral of~\eqref{eqn:hclwe-def} (or equivalently, of~\eqref{eqn:hclwe-mixture-def}) over all $\by$ is
\begin{align}\label{eqn:hclwe-def-normalization}
    Z = \frac{\beta}{\sqrt{\beta^2+\gamma^2}} \cdot \rho\Bigg(\frac{1}{\sqrt{\beta^2+\gamma^2}}\mathbb{Z}\Bigg) \; .
\end{align}
This is easy to see since the integral over $\by$ of the product of the last two $\rho$ terms in~\eqref{eqn:hclwe-mixture-def} is $\beta/\sqrt{\beta^2+\gamma^2}$ independently of $k$.

\begin{definition} For parameters $\beta, \gamma > 0$, the average-case decision problem $\clwe_{\beta, \gamma}$ is to distinguish the following two distributions over $\mathbb{R}^n \times \mathbb{T}$: (1) the CLWE distribution $A_{\bw, \beta, \gamma}$ for some uniformly random unit vector $\bw \in \mathbb{R}^n$ (which is fixed for all samples), or (2) $D_{\mathbb{R}^n} \times U$.
\end{definition}

\begin{definition} For parameters $\beta, \gamma > 0$, the average-case decision problem $\hclwe_{\beta, \gamma}$ is to distinguish the following two distributions over $\mathbb{R}^n$: (1) the homogeneous CLWE distribution $H_{\bw, \beta, \gamma}$ for some uniformly random unit vector $\bw \in \mathbb{R}^n$ (which is fixed for all samples), or (2) $D_{\mathbb{R}^n}$.
\end{definition}

Note that $\clwe_{\beta,\gamma}$ and $\hclwe_{\beta,\gamma}$ are defined as average-case problems. We could have equally well defined them to be worst-case problems, requiring the algorithm to distinguish the distributions for \emph{all} hidden directions $\bw \in \mathbb{R}^n$. The following claim shows that the two formulations are equivalent.

\begin{claim} \label{claim:ic-worst-to-ic}
For any $\beta, \gamma > 0$,
there is a polynomial-time reduction from worst-case $\clwe_{\beta,\gamma}$ to (average-case) $\clwe_{\beta,\gamma}$.
\end{claim}
\begin{proof}
Given CLWE samples $\{(\by_i,z_i)\}_{i=1}^{\poly(n)}$ from $A_{\bw,\beta,\gamma}$, we apply a random rotation $\bR$, giving us samples of the form $\{(\bR\by_i,z_i\}_{i=1}^{\poly(n)}$. Since the standard Gaussian is rotationally invariant and $\langle \by, \bw \rangle = \langle \bR \by, \bR^T\bw \rangle$, the rotated CLWE samples are distributed according to $A_{\bR^T\bw,\beta,\gamma}$. Since $\bR$ is a random rotation, the random direction $\bR^T \bw$ is uniformly distributed on the sphere.
\end{proof}

\section{Hardness of CLWE}
\label{section:clwe-hardness}

\subsection{Background and overview}
\label{section:clwe-background}
In this section, we give an overview of the quantum reduction from worst-case lattice problems to CLWE. Our goal is to show that we can efficiently solve worst-case lattice problems, in particular GapSVP and SIVP, using an oracle for $\clwe$ (and with quantum computation). We first state our main theorem, which was stated informally as Theorem~\ref{thm:main-informal} in the introduction. 

\begin{theorem}
\label{thm:clwe-intro}
Let $\beta = \beta(n) \in (0,1)$ and $\gamma = \gamma(n) \geq 2\sqrt{n}$ be such that $\gamma/\beta$ is polynomially bounded. Then there is a polynomial-time quantum reduction from $\dgs_{2\sqrt{n}\eta_\eps(L)/\beta}$ to $\clwe_{\beta,\gamma}$.
\end{theorem}

Using standard reductions from GapSVP and SIVP to DGS  (see, e.g.,~\cite[Section 3.3]{regev2005lwe}), our main theorem immediately implies the following corollary.
\begin{corollary}
Let $\beta = \beta(n) \in (0,1)$ and $\gamma = \gamma(n) \geq 2\sqrt{n}$ such that $\gamma/\beta$ is polynomially bounded. Then, there is a polynomial-time quantum reduction from $\mathrm{SIVP}_\alpha$ and $\mathrm{GapSVP}_\alpha$ to $\clwe_{\beta,\gamma}$ for some $\alpha = \tilde{O}(n/\beta)$.
\end{corollary}

Based on previous work, to prove Theorem~\ref{thm:clwe-intro}, it suffices to prove the following lemma, which is the goal of this section. 

\begin{lemma}
\label{lem:bdddgs-to-clwe}
Let $\beta=\beta(n) \in (0,1)$ and $\gamma=\gamma(n) \geq 2\sqrt{n}$ such that $q = \gamma/\beta$ is polynomially bounded. There exists a probabilistic polynomial-time (classical) algorithm with access to an oracle that solves $\clwe_{\beta,\gamma}$, that takes as input a lattice $L \subset \mathbb{R}^n$, parameters $\beta, \gamma$, and $r \geq 2q \cdot \eta_{\eps}(L)$, and $\poly(n)$ many samples from the discrete Gaussian distribution $D_{L,r_i}$ for $\poly(n)$ parameters $r_i \geq r$ and solves $\bdd_{L^*,d}$ for $d = \gamma/(\sqrt{2}r)$.
\end{lemma}

In other words, we can implement an oracle for $\bdd_{L^*,\gamma/(\sqrt{2}r)}$ using polynomially many discrete Gaussian samples and the CLWE oracle as a sub-routine.
The proof of Lemma~\ref{lem:bdddgs-to-clwe} will be given in Section~\ref{section:clwe-from-bdddgs} 
(which is the novel contribution) and Section~\ref{section:solve-bdd-with-clwe}
(which mainly follows~\cite{peikert2017ringlwe}).

In the rest of this subsection, we briefly explain how Theorem~\ref{thm:clwe-intro} follows from Lemma~\ref{lem:bdddgs-to-clwe}. 
This derivation is already implicit in past work~\cite{peikert2017ringlwe,regev2005lwe}, and is included here mainly for completeness. 
Readers familiar with the reduction may skip directly to Section~\ref{section:clwe-from-bdddgs}. 

The basic idea is to start with samples from a very wide discrete Gaussian (which can be efficiently sampled) and then iteratively sample from narrower discrete Gaussians, until eventually we end up with short discrete Gaussian samples, as required (see Figure~\ref{fig:lwe-diagram}). Each iteration consists of two steps: the first classical step is given by Lemma~\ref{lem:bdddgs-to-clwe}, allowing us to solve BDD on the dual lattice; the second step is quantum and is given in Lemma~\ref{lem:reg05quantumstep} below, which shows that solving BDD leads to sampling from narrower discrete Gaussian.

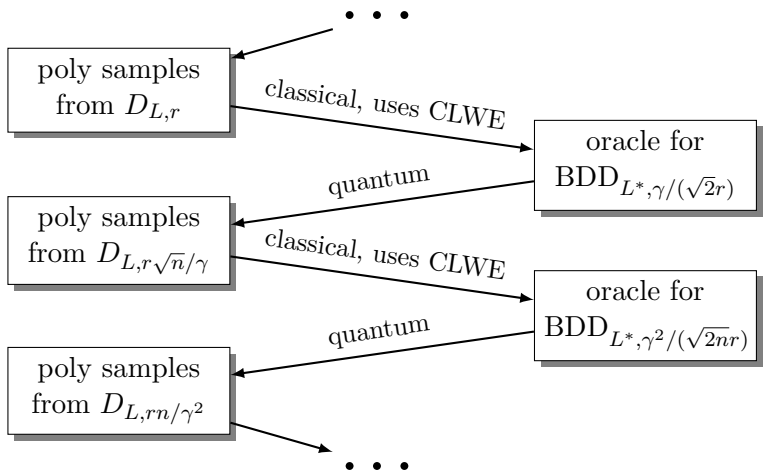
\begin{figure}[ht]
\centering
\input{plots/clwe-iteration.tex}
\caption{Two iterations of the reduction.}
\label{fig:lwe-diagram}
\end{figure}

\begin{lemma}[{\cite[Lemma 3.14]{regev2005lwe}}]
\label{lem:reg05quantumstep}
There exists an efficient quantum algorithm that, given any $n$-dimensional lattice $L$, a number $d < \lambda_1(L^*)/2$, and an oracle that solves $\bdd_{L^*,d}$, outputs a sample from $D_{L,\sqrt{n}/(\sqrt{2}d)}$.
\end{lemma}

Similar to~\cite{peikert2017ringlwe}, there is a subtle requirement in Lemma~\ref{lem:bdddgs-to-clwe} that we need discrete Gaussian samples from several different parameters $r' \geq r$. However, this is a non-issue since an oracle for $\bdd_{L^*,\gamma/(\sqrt{2}r)}$ also solves $\bdd_{L^*,\gamma/(\sqrt{2}r')}$ for any $r' \ge r$, so Lemma~\ref{lem:reg05quantumstep} in fact allows us to efficiently sample from $D_{L,r'\sqrt{n}/\gamma}$ for any $r' \ge r$.

\subsection{CLWE samples from BDD}
\label{section:clwe-from-bdddgs}

In this subsection we prove Lemma~\ref{lem:bdd-to-clwe}, showing how to generate CLWE samples from the given BDD instance using discrete Gaussian samples. 
In the next subsection we will show how to solve the BDD instance by applying the decision CLWE oracle to these samples, thereby completing the proof of Lemma~\ref{lem:bdddgs-to-clwe}.

\begin{lemma}
\label{lem:bdd-to-clwe}
There is an efficient algorithm that takes as input an $n$-dimensional lattice $L$, a vector $\bw + \bu$ where $\bu \in L^*$, reals $r, s_1, s_2 > 0$ such that $rs_1/\sqrt{\|\bw\|^2(r s_1/s_2)^2+t^2} \geq \eta_\eps(L)$ for some $\eps < \frac{1}{2}$ and $t = \sqrt{r^2+s_1^2}$, and samples from $D_{L,r}$, and outputs samples that are within statistical distance $8 \eps$ of the CLWE distribution $A_{\bw', \beta, \gamma}$ for $\bw' = \bw/\|\bw\|$, $\beta = \|\bw\|\sqrt{(rs_1/t)^2+(s_2/\|\bw\|)^2}$ and $\gamma = \|\bw\|r^2/t$.
\end{lemma}
\begin{proof}
We start by describing the algorithm. For each $\bx$ from the given samples from $D_{L,r}$, do the following. First, take the inner product $\langle \bx, \bw + \bu \rangle$, which gives us
\begin{align*}
    \langle \bx, \bw + \bu \rangle &= \langle \bx, \bw \rangle \bmod 1
    \; .
\end{align*}
Appending this inner product modulo 1 to the sample $\bx$, we get $(\bx, \langle \bx, \bw \rangle \bmod 1)$. 
Next, we ``smooth out" the lattice structure of $\bx$ by adding Gaussian noise $\bv \sim D_{\mathbb{R}^n,s_1}$ to $\bx$ and $e \sim D_{\mathbb{R},s_2}$ to $\langle \bx, \bw \rangle$ (modulo 1). Then, we have
\begin{align}
    (\bx + \bv, (\langle \bx, \bw \rangle + e) \bmod 1) \label{eqn:clwe-sample-raw}\; .
\end{align}
Finally, we normalize the first component by $t$ so that its marginal distribution has unit width, giving us
\begin{align}
    ((\bx + \bv)/t,(\langle \bx, \bw \rangle + e) \bmod 1) \label{eqn:clwe-sample-raw-normalized}\;,
\end{align}
which the algorithm outputs.

Our goal is to show that the distribution of \eqref{eqn:clwe-sample-raw-normalized} is within statistical distance $8\eps$ of the CLWE distribution $A_{\bw',\beta,\gamma}$, given by 
\begin{align*}
    (\by', (\gamma \langle \by', \bw' \rangle + e') \bmod 1) \; ,
\end{align*}
where $\by' \sim D_{\mathbb{R}^n}$ and $e' \sim D_{\mathbb{R},\beta}$.
Because applying a function cannot increase statistical distance (specifically, dividing the first component by $t$ and taking mod $1$ of the second), it suffices to show that the distribution of
\begin{align}
    (\bx + \bv, \langle \bx, \bw \rangle + e) \label{eqn:clwe-sample-1}\; ,
\end{align}
is within statistical distance $8\eps$ of that of
\begin{align}
    (\by, (r/t)^2 \langle \by, \bw \rangle + e') \label{eqn:clwe-sample-2}\; ,
\end{align}
where $\by \sim D_{\mathbb{R}^n,t}$ and $e' \sim D_{\mathbb{R},\beta}$. First, observe that by Lemma~\ref{lem:smoothing-gaussian}, the statistical distance between the marginals on the first component (i.e., between $\bx +\bv$ and $\by$) is at most $4\eps$. It is therefore sufficient to bound the statistical distance between the second components conditioned on any fixed value $\overline{\by}$ of the first component. 
Conditioned on the first component being $\overline{\by}$, the second component in~\eqref{eqn:clwe-sample-1} has the same distribution as
\begin{align}
\langle \bx + \bh , \bw \rangle \label{eqn:clwe-sample-3}
\end{align}
where $\bh \sim D_{\mathbb{R}^n,s_2/\|\bw\|}$, 
and the second component in~\eqref{eqn:clwe-sample-2} has the same distribution as
\begin{align}
\langle (r/t)^2 \overline{\by} + \bh' , \bw \rangle \label{eqn:clwe-sample-4}
\end{align}
where $\bh' \sim D_{\mathbb{R}^n,\beta/\|\bw\|}$.

By Claim~\ref{clm:lattice_conditional} below, conditioned on $\bx+\bv = \overline{\by}$, the distribution of $\bx$ is
$(r/t)^2\overline{\by} + D_{L-(r/t)^2\overline{\by}, rs_1/t}$. Therefore, by Lemma~\ref{lem:smoothing-gaussian}, the conditional distribution of $\bx + \bh$ given $\bx+\bv=\overline{\by}$ is within statistical distance $4 \eps$ of that of $(r/t)^2\overline{\by} + \bh'$. Since statistical distance cannot increase by applying a function (inner product with $\bw$ in this case), \eqref{eqn:clwe-sample-3} is within statistical distance $4\eps$ of \eqref{eqn:clwe-sample-4}. Hence, the distribution of \eqref{eqn:clwe-sample-1} is within statistical distance $8\eps$ of that of \eqref{eqn:clwe-sample-2}.
\end{proof}

\begin{claim}
\label{clm:lattice_conditional}
Let $\by = \bx + \bv$, where $\bx \sim D_{L,r}$ and $\bv \sim D_{\mathbb{R}^n,s}$. Then, the conditional distribution of $\bx$ given $\by = \overline{\by}$ is $(r/t)^2\overline{\by} + D_{L-(r/t)^2\overline{\by}, rs/t}$ where $t = \sqrt{r^2+s^2}$.
\end{claim}
\begin{proof}
Observe that $\bx$ conditioned on $\by = \overline{\by}$ is a discrete random variable supported on $L$.
The probability of $\bx$ given $\by = \overline{\by}$ is proportional to
\begin{align*}
    \rho_r(\bx) \cdot \rho_s(\overline{\by}-\bx) = \rho_t(\overline{\by}) \cdot \rho_{rs/t}(\bx-(r/t)^2\overline{\by}) \propto \rho_{rs/t}(\bx-(r/t)^2\overline{\by})
    \; ,
\end{align*}
where the equality follows from Claim~\ref{claim:complete-squares}. Hence, the conditional distribution of $\bx-(r/t)^2\by$ given $\by = \overline{\by}$ is $D_{L-(r/t)^2\overline{\by}, rs/t}$. 
\end{proof}

\subsection{Solving BDD with the CLWE oracle}
\label{section:solve-bdd-with-clwe}
In this subsection, we complete the proof of Lemma~\ref{lem:bdddgs-to-clwe}. We first give some necessary background on the Oracle Hidden Center Problem (OHCP)~\cite{peikert2017ringlwe}. The problem asks one to search for a ``hidden center" $\bw^*$ using a decision oracle whose acceptance probability depends only on the distance to $\bw^*$. The problem's precise statement is as follows.

\begin{definition}[OHCP] 
\label{definition:ohcp}
For parameters $\eps, \delta \in [0,1)$ and $\zeta \geq 1$, the $(\eps, \delta, \zeta)$-\emph{OHCP} is an approximate search problem that tries to find the ``hidden" center $\bw^*$. Given a scale parameter $d > 0$ and access to a randomized oracle $\mathcal{O} : \mathbb{R}^n \times \mathbb{R}^{\geq 0} \rightarrow \{0,1\}$ such that its acceptance probability $p(\bw,t)$ only depends on $\exp(t)\|\bw-\bw^*\|$ for some (unknown) ``hidden center" $\bw^* \in \mathbb{R}^n$ with $\delta d \leq \|\bw^*\| \leq d$ and for any $\bw \in \mathbb{R}^n$ with $\|\bw-\bw^*\| \leq \zeta d$, the goal is to output $\bw$ s.t.\ $\|\bw-\bw^*\| \leq \eps d$.
\end{definition}

Notice that OHCP corresponds to our problem since we want to solve BDD, which is equivalent to finding the ``hidden" offset vector $\bw^*$, using a decision oracle for $\clwe_{\beta, \gamma}$. The acceptance probability of the $\clwe_{\beta,\gamma}$ oracle will depend on the distance between our guess $\bw$ and the true offset $\bw^*$. For OHCP, we have the following result from~\cite{peikert2017ringlwe}.

\begin{lemma}[\cite{peikert2017ringlwe}, Proposition 4.4]
\label{lem:ohcp}
There is a poly$(\kappa, n)$-time algorithm that takes as input a confidence parameter $\kappa \geq 20 \log(n+1)$ (and the scale parameter $d > 0$) and solves $(\exp(-\kappa), \exp(-\kappa), 1+1/\kappa)$-OHCP in dimension $n$ except with probability $\exp(-\kappa)$, provided that the oracle $\cO$ corresponding to the OHCP instance satisfies the following conditions. For some $p(\infty) \in [0, 1]$ and $t^* \ge 0$,
\begin{enumerate}
    \item $p(\bzero, t^*) - p(\infty) \geq 1/\kappa$;
    \item $|p(\bzero, t) - p(\infty)| \leq 2 \exp(-t/\kappa)$ for any $t \geq 0$; and
    \item $p(\bw, t)$ is $\kappa$-Lipschitz in $t$ for any $\bw \in \mathbb{R}^n$ such that $\|\bw\| \leq (1+1/\kappa)d$ \;.
\end{enumerate}
Furthermore, each of the algorithm's oracle calls takes the form $\cO(\cdot,i\Delta)$ for some $\Delta < 1$ that depends only on $\kappa$ and $n$ and $0 \leq i \leq \poly(\kappa,n)$.
\end{lemma}

The main idea in the proof of Lemma~\ref{lem:ohcp} is performing a guided random walk with advice from the decision oracle $\mathcal{O}$. The decision oracle $\mathcal{O}$ rejects a random step with high probability if it increases the distance $\|\bw - \bw^*\|$. Moreover, there is non-negligible probability of decreasing the distance by a factor $\exp(1/n)$ unless $\log \|\bw-\bw^*\| \leq -\kappa$. Hence, with sufficiently many steps, the random walk will reach $\widehat{\bw}$, a guess of the hidden center, which is within $\exp(-\kappa)$ distance to $\bw^*$ with high probability.

Our goal is to show that we can construct an oracle $\mathcal{O}$ satisfying the above conditions using an oracle for $\clwe_{\beta, \gamma}$. Then, it follows from Lemma~\ref{lem:ohcp} that BDD with discrete Gaussian samples can be solved using an oracle for CLWE. We first state some lemmas useful for our proof. Lemma~\ref{lem:closest-plane} is Babai's closest plane algorithm and Lemma~\ref{lem:tvbound} is an upper bound on the statistical distance between two one-dimensional Gaussian distributions.

\begin{lemma}[\cite{lenstra1982lll,babai1986cvp}]
\label{lem:closest-plane}
For any $n$-dimensional lattice $L$, there is an efficient algorithm that solves $\bdd_{L,d}$ for $d = 2^{-n/2}\cdot \lambda_1(L)$.
\end{lemma}

\begin{lemma}[{\cite[Theorem 1.3]{devroye2018tv}}] For all $\mu_1, \mu_2 \in  \mathbb{R}$, and $\sigma_1, \sigma_2 > 0$, we have
\label{lem:tvbound}
\begin{align*}
    \Delta\big(\cN(\mu_1,\sigma_1),\cN(\mu_2,\sigma_2)\big) \leq \frac{3|\sigma_1^2-\sigma_2^2|}{2\max(\sigma_1^2,\sigma_2^2)}+\frac{|\mu_1-\mu_2|}{2\max(\sigma_1,\sigma_2)}
    \; ,
\end{align*}
where $\cN(\mu,\sigma)$ denotes the Gaussian distribution with mean $\mu$ and standard deviation $\sigma$.
\end{lemma}

Now, we prove Lemma~\ref{lem:bdddgs-to-clwe}, restated below.

{
\def\thetheorem{\ref{lem:bdddgs-to-clwe}}
\begin{lemma}
Let $\beta=\beta(n) \in (0,1)$ and $\gamma=\gamma(n) \geq 2\sqrt{n}$ such that $q = \gamma/\beta$ is polynomially bounded. There exists a probabilistic polynomial-time (classical) algorithm with access to an oracle that solves $\clwe_{\beta,\gamma}$, that takes as input a lattice $L \subset \mathbb{R}^n$, parameters $\beta, \gamma$, and $r \geq 2q \cdot \eta_{\eps}(L)$, and $\poly(n)$ many samples from the discrete Gaussian distribution $D_{L,r_i}$ for $\poly(n)$ parameters $r_i \geq r$ and solves $\bdd_{L^*,d}$ for $d = \gamma/(\sqrt{2}r)$.
\end{lemma}
\addtocounter{theorem}{-1}
}

\begin{proof}
Let $d' = (1-1/(2n))\cdot d$. By~\cite[Corollary 2]{lyubashevsky2009bdd}, it suffices to solve $\bdd_{L^*,d'}$. 
Let $\kappa = \poly(n)$ with $\kappa \geq 8qn\ell$ be such that the advantage of our $\clwe_{\beta, \gamma}$ oracle is at least $1/\kappa$, where $\ell \geq 1$ is the number of samples required by the oracle.

Given as input a lattice $L \subset \mathbb{R}^n$, a parameter $r \geq 2q \cdot \eta_{\eps}(L)$, samples from $D_{L,r_i}$ for $1 \leq i \leq \poly(n)$, and a BDD instance $\bw^* + \bu$ where $\bu \in L^*$ and $\|\bw^*\| \leq d'$, we want to recover $\bw^*$. Without loss of generality, we can assume that $\|\bw^*\| \geq \exp(-n/2)\cdot \lambda_1(L^*) \geq (2q/r)\cdot \exp(-n/2)$ (Lemma~\ref{lem:smoothing-lb}), since we can otherwise find $\bw^*$ efficiently using Babai's closest plane algorithm (Lemma~\ref{lem:closest-plane}).

We will use the CLWE oracle to simulate an oracle $\mathcal{O}: \mathbb{R}^n \times \mathbb{R}^{\ge 0} \rightarrow \{0,1\}$ such that the probability that $\mathcal{O}(\bw,t)$ outputs 1 (``accepts") only depends on $\exp(t)\|\bw-\bw^*\|$. Our oracle $\cO$ corresponds to the oracle in Definition~\ref{definition:ohcp} with $\bw^*$ as the ``hidden center". We will use Lemma~\ref{lem:ohcp} to find $\bw^*$.

On input $(\bw, t)$, our oracle $\mathcal{O}$ receives $\ell$ independent samples from $D_{L,\exp(t)r}$. Then, we generate CLWE samples using the procedure from Lemma~\ref{lem:bdd-to-clwe}. The procedure takes as input these $\ell$ samples, the vector $\bu + \bw^* - \bw$ where $\bu \in L^*$, and parameters $\exp(t) r, \exp(t) s_1, s_2$. Our choice of $s_1$ and $s_2$ will be specified below. Note that the CLWE oracle requires the ``hidden direction" $(\bw-\bw^*)/\|\bw-\bw^*\|$ to be uniformly distributed on the unit sphere. To this end, we apply the worst-to-average case reduction from Claim~\ref{claim:ic-worst-to-ic}. Let $S_{\bw, t}$ be the resulting CLWE distribution. Our oracle $\mathcal{O}$ then calls the $\clwe_{\beta,\gamma}$ oracle on $S_{\bw,t}^\ell$ and outputs 1 if and only if it accepts.

Using the oracle $\mathcal{O}$, we can run the procedure from Lemma~\ref{lem:ohcp} with confidence parameter $\kappa$ and scale parameter $d'$. The output of this procedure will be some approximation $\what{\bw}$ to the oracle's ``hidden center" with the guarantee that $\|\what{\bw}-\bw^*\| \leq \exp(-\kappa)d'$. Finally, running Babai's algorithm on the vector $\bu+\bw^*-\what{\bw}$ will give us $\bw^*$ exactly since
\begin{align*}
    \|\what{\bw}-\bw^*\| \leq \exp(-\kappa)d \leq \beta\exp(-\kappa)/\eta_\eps(L) \leq 2^{-n}\lambda_1(L^*)
    \; ,
\end{align*}
where the last inequality is from Lemma~\ref{lem:smoothing-dual}.

The running time of the above procedure is clearly polynomial in $n$. It remains to check that our oracle $\mathcal{O}$ (1) is a valid instance of $(\exp(-\kappa),\exp(-\kappa),1+1/\kappa)$-OHCP with hidden center $\bw^*$ and (2) satisfies all the conditions of Lemma~\ref{lem:ohcp}. First, note that $S_{\bw, t}$ will be negligibly close in statistical distance to the CLWE distribution with parameters
\begin{align*}
    \beta' &= \sqrt{(\exp(t)\|\bw-\bw^*\|)^2s_1'^2+s_2^2}
    \; , \\
    \gamma' &= \exp(t)\|\bw-\bw^*\|r'
    \; ,
\end{align*}
where $r' = r^2/\sqrt{r^2+s_1^2}$ and $s_1' = rs_1/\sqrt{r^2+s_1^2}$ as long as $r,s_1,s_2$ satisfy the conditions of Lemma~\ref{lem:bdd-to-clwe}. Then, we set $s_1 = r/(\sqrt{2}q)$ and choose $s_2$ such that
\begin{align*}
    s_2^2 = {\beta}^2 - (s_1'/r')^2{\gamma}^2 = {\beta}^2 - (s_1/r)^2{\gamma}^2 = {\beta}^2/2
    \; .
\end{align*}
    
Lemma~\ref{lem:bdd-to-clwe} requires $rs_1/\sqrt{r^2\|\bw-\bw^*\|^2(s_1/s_2)^2+r^2+s_1^2} \geq \eta_{\eps}(L)$. We know that $r \geq 2q\cdot \eta_{\eps}(L)$ and $s_1 \geq \sqrt{2}\cdot \eta_{\eps}(L)$, so it remains to determine a sufficient condition for the aforementioned inequality. Observe that for any $\bw$ such that $\|\bw-\bw^*\| \leq d$, the condition $s_2 \geq 2d\cdot\eta_\eps(L)$ is sufficient. Since $r \geq 2(\gamma/\beta)\cdot \eta_{\eps}(L)$, this translates to $s_2 \geq \beta/(\sqrt{2})$. Hence, the transformation from Lemma~\ref{lem:bdd-to-clwe} will output samples negligibly close to CLWE samples for our choice of $s_1$ and $s_2$ as long as $\|\bw-\bw^*\| \leq d$ (beyond the BDD distance bound $d'$).

Since $S_{\bw,t}$ is negligibly close to the CLWE distribution, the acceptance probability $p(\bw,t)$ of $\mathcal{O}$ only depends on $\exp(t)\|\bw-\bw^*\|$. Moreover, by assumption $\|\bw^*\| \geq \exp(-n/2) \cdot (2q/r) \geq \exp(-\kappa)d'$. Hence, $\mathcal{O}, \kappa, d'$ correspond to a valid instance of $(\exp(-\kappa),\exp(-\kappa),1+1/\kappa)$-OHCP with ``hidden center" $\bw^*$.

Next, we show that $p(\bw,t)$ of $\mathcal{O}$ satisfies all three conditions of Lemma~\ref{lem:ohcp} with $p(\infty)$ taken to be the acceptance probability of the CLWE oracle on samples from $D_{\mathbb{R}^n} \times U$. 
Item~1 of Lemma~\ref{lem:ohcp} follows from our assumption that our $\clwe_{\beta,\gamma}$ oracle has advantage $1/\kappa$, and by our choice of $r$, $s_1$, and $s_2$, when $t^* = \log(\gamma/(\|\bw^*\|r')) > \log(\sqrt{2})$, the generated CLWE samples satisfy $\gamma'(t^*) =  \gamma$ and $\beta'(t^*) = \beta$. Hence, $p(\bzero,t^*) - p(\infty) \geq 1/\kappa$.

We now show that Item~2 holds, which states that $|p(\bzero,t)-p(\infty)| \leq 2 \exp(-t/\kappa)$ for any $t > 0$. We will show that $S_{\bzero, t}$ converges exponentially fast to $D_{\mathbb{R}^n} \times U$ in statistical distance. Let $f(\by,z)$ be the probability density of $S_{\bzero, t}$. Then,
\begin{align*}
    \Delta(S_{\bzero,t},D_{\mathbb{R}^n}\times U) &= \frac{1}{2}\int |f(z|\by)-U(z)|\rho(\by)d\by dz \\
    &= \frac{1}{2} \int \Big(\int |f(z|\by)-U(z)|dz\Big)\rho(\by) d\by
    \; .
\end{align*}
Hence, it suffices to show that the conditional density of $z$ given $\by$ for $S_{\bzero,t}$ converges exponentially fast to the uniform distribution on $\mathbb{T}$. Notice that the conditional distribution of $z$ given $\by$ is the Gaussian distribution with width parameter $\beta' \geq \exp(t)\|\bw^*\|r/(2q) \geq \exp(t-n/2)$, where we have used our assumption that $\|\bw^*\| \geq (2q/r)\cdot \exp(-n/2)$. By Lemma~\ref{lem:smoothing-dual} applied to $\mathbb{Z}$, we know that $\beta'$ is larger than $\eta_{\eps}(\mathbb{Z})$ for $\eps = \exp(-\exp(2t-n))$. Hence, one sample from this conditional distribution is within statistical distance $\eps$ of the uniform distribution by Lemma~\ref{lem:smoothing-uniform}. By the triangle inequality applied to $\ell$ samples,
\begin{align*}
    \Delta\Big(S_{\bzero, t}^\ell, (D_{\mathbb{R}^n} \times U)^\ell\Big) \leq \min(1, \ell \exp(-\exp(2t-n))) \leq 2\exp(-t/\kappa)
    \; ,
\end{align*}
where in the last inequality, we use the the fact that we can choose $\kappa$ to be such that $2\exp(-t/\kappa) \geq 1$ unless $t \geq \kappa/2$. And when $t \geq \kappa/2 \geq 4qn\ell$, we have $\ell \exp(-\exp(2t-n)) \ll \exp(-t/\kappa)$.

It remains to verify Item~3, which states that $p(\bw, t)$ is $\kappa$-Lipschitz in $t$ for any $\|\bw\| \leq (1+1/\kappa)d' \leq d$. We show this by bounding the statistical distance between $S_{\bw,t_1}$ and $S_{\bw,t_2}$ for $t_1 \geq t_2$. With a slight abuse in notation, let $f_{t_i}(\by,z)$ be the probability density of $S_{\bw,t_i}$ and let $(\beta_i, \gamma_i)$ be the corresponding CLWE distribution parameters. For simplicity, also denote the hidden direction by $\bw' = (\bw-\bw^*)/\|\bw-\bw^*\|$. Then,

\begin{align}
    \Delta(f_{t_1}, f_{t_2})
    &= \frac{1}{2} 
    \int \Big(\int |f_{t_1}(z|\by)-f_{t_2}(z|\by)|dz\Big) \rho(\by)d\by \nonumber \\
    &= \int \Delta\Big(\cN(\gamma_1\langle\by,\bw'\rangle,\beta_1/\sqrt{2\pi}),\cN(\gamma_2\langle\by,\bw'\rangle,\beta_2/\sqrt{2\pi})\Big) \rho(\by)d\by \nonumber \\
    &\leq \frac{1}{2} \int \Big(3(1-(\beta_2/\beta_1)^2) + \sqrt{2\pi}(\gamma_1-\gamma_2)/\beta_1\cdot|\langle \by, \bw' \rangle|\Big)\cdot \rho(\by)d\by \label{eqn:devroye-tv}\\
    &\leq \E_{\by \sim \rho}[M(\by)] 
    \cdot \Big(1-\exp(-2(t_1-t_2))\Big) \text{ where } M(\by) 
    = \frac{1}{2}\Big(3+2\sqrt{\pi} q \cdot|\langle \by, \bw' \rangle|\Big) \nonumber \\
    &\leq \E_{\by \sim \rho}[M(\by)] \cdot 2(t_1-t_2)  \label{eqn:linear-bound} \\
    &\leq (\kappa/\ell)\cdot (t_1-t_2) \label{eqn:exp-half-gaussian}
    \; ,
\end{align}
where \eqref{eqn:devroye-tv} follows from Lemma~\ref{lem:tvbound}, \eqref{eqn:linear-bound} uses the fact that $1-\exp(-2(t_1-t_2)) \leq 2(t_1-t_2)$, and \eqref{eqn:exp-half-gaussian} uses the fact that $\E_{\by \sim \rho}[M(\by)] \leq 4q \leq \kappa/(2\ell)$. Using the triangle inequality over $\ell$ samples, the statistical distance between $S_{\bw,t_1}^\ell$ and $S_{\bw,t_2}^\ell$ is at most
\begin{align*}
    \min(1,\ell\cdot(\kappa/\ell)(t_1-t_2)) \leq \kappa(t_1-t_2)
    \; .
\end{align*}
Therefore, $p(\bw,t)$ is $\kappa$-Lipschitz in $t$.
\end{proof}

\section{Hardness of Homogeneous CLWE}
\label{section:hc}

In this section, we show the hardness of homogeneous CLWE by reducing from CLWE, whose hardness was established in the previous section.
The main step of the reduction is to transform CLWE samples to homogeneous CLWE samples using rejection sampling (Lemma~\ref{lem:ic-to-hc}).

Consider the samples $(\by, z) \sim A_{\bw,\beta,\gamma}$ in $\clwe_{\beta,\gamma}$. If we condition $\by$ on $z = 0 \pmod{1}$ then we get exactly samples $\by \sim H_{\bw,\beta,\gamma}$ for $\hclwe_{\beta,\gamma}$. However, this approach is impractical as $z = 0 \pmod{1}$ happens with probability 0. Instead we condition $\by$ on $z \approx 0 \pmod{1}$ somehow. One can imagine that the resulting samples $\by$ will still have a ``wavy" probability density in the direction of $\bw$ with spacing $1/\gamma$, which accords with the picture of homogeneous CLWE. To avoid throwing away too many samples, we will do rejection sampling with some small ``window" $\delta = 1/\poly(n)$. Formally, we have the following lemma.
\begin{lemma}
\label{lem:ic-to-hc}
There is a $\poly(n, 1/\delta)$-time probabilistic algorithm that takes as input a parameter $\delta \in (0,1)$ and samples from $A_{\bw,\beta,\gamma}$, and outputs samples from $H_{\bw,\sqrt{\beta^2+\delta^2},\gamma}$.
\end{lemma}
\begin{proof}

Without loss of generality assume that $\bw = \be_1$.
By definition, the probability density of sample $(\by, z) \sim A_{\bw,\beta,\gamma}$ is
\begin{align*}
    p(\by, z) = \frac{1}{\beta}\cdot \rho(\by) \cdot \sum_{k \in \mathbb{Z}} \rho_\beta(z+k-\gamma y_1)
    \; .
\end{align*}
Let $g : \mathbb{T} \to [0,1]$ be the function $g(z) = g_0(z) / M$, where $g_0(z) = \sum_{k \in \mathbb{Z}} \rho_\delta(z+k)$ 
and $M = \sup_{z \in \mathbb{T}} g_0(z)$.
We perform rejection sampling on the samples $(\by, z)$ with acceptance probability $\Pr[\mathrm{accept} | \by, z] = g(z)$.
We remark that $g(z)$ is efficiently computable (see~\cite[Section 5.2]{BrakerskiLPRS13}).
The probability density of outputting $\by$ and accept is 
\begin{align*}
    \int_\mathbb{T}  p(\by, z) g(z) d z 
    &= \frac{\rho(\by)}{\beta M} \cdot \int_\mathbb{T} \sum_{k_1, k_2 \in \mathbb{Z}} \rho_\beta(z+k_1-\gamma y_1) \rho_\delta(z+k_2) d z \\
    &= \frac{\rho(\by)}{\beta M} \cdot \int_\mathbb{T} \sum_{k, k_2 \in \mathbb{Z}} \rho_{\sqrt{\beta^2+\delta^2}}(k-\gamma y_1) \rho_{\beta\delta/\sqrt{\beta^2+\delta^2}} \Bigl( z+k_2+\frac{\delta^2 (k-\gamma y_1)}{\beta^2+\delta^2} \Bigr) d z \\
    &= \frac{\delta}{M \sqrt{\beta^2+\delta^2}} \cdot \rho(\by) \cdot \sum_{k \in \mathbb{Z}} \rho_{\sqrt{\beta^2+\delta^2}}(k-\gamma y_1)
    \; ,
\end{align*}
where the second equality follows from Claim~\ref{claim:complete-squares}.
This shows that the conditional distribution of $\by$ upon acceptance is indeed $H_{\be_1,\sqrt{\beta^2+\delta^2},\gamma}$.
Moreover, a byproduct of this calculation is that the expected acceptance probability is $\Pr[\mathrm{accept}] = Z \delta / (M \sqrt{\beta^2+\delta^2})$, where, according to Eq.~\eqref{eqn:hclwe-def-normalization},
\begin{align*}
    Z
    &= \sqrt\frac{\beta^2+\delta^2}{\beta^2+\delta^2+\gamma^2} \cdot \rho_{\sqrt{\beta^2+\delta^2+\gamma^2}}(\mathbb{Z}) \\
    &= \sqrt{\beta^2+\delta^2} \cdot \rho_{1/\sqrt{\beta^2+\delta^2+\gamma^2}}(\mathbb{Z}) \\
    &\ge \sqrt{\beta^2+\delta^2}
    \; ,
\end{align*}
and the second equality uses Lemma~\ref{lem:poisson-sum}.
Observe that
\begin{align*}
    g_0(z) &= \sum_{k \in \mathbb{Z}} \rho_\delta(z+k) \\
    &\leq 2 \cdot \sum_{k = 0}^\infty \rho_\delta(k) \\
    &< 2 \cdot \sum_{k = 0}^\infty \exp(-\pi k)
    < 4
\end{align*}
since $\delta < 1$, implying that $M \le 4$.
Therefore, $\Pr[\mathrm{accept}] \ge \delta/4$, and so the rejection sampling procedure has $\poly(n, 1/\delta)$ expected running time.
\end{proof}
The above lemma reduces CLWE to homogeneous CLWE with slightly worse parameters. Hence, homogeneous CLWE is as hard as CLWE. 
Specifically, combining Theorem~\ref{thm:clwe-intro} (with $\beta$ taken to be $\beta/\sqrt{2}$) and Lemma~\ref{lem:ic-to-hc} (with $\delta$ also taken to be $\beta/\sqrt{2}$), we obtain the following corollary.

\begin{corollary}
\label{cor:hc}
For any $\beta = \beta(n) \in (0,1)$ and $\gamma = \gamma(n) \geq 2\sqrt{n}$ such that $\gamma/\beta$ is polynomially bounded,
there is a polynomial-time quantum reduction from $\dgs_{2\sqrt{2n}\eta_\eps(L)/\beta}$ to $\hclwe_{\beta,\gamma}$.
\end{corollary}

\section{Hardness of Density Estimation for Gaussian Mixtures}
\label{section:mixture-hardness}
In this section, we prove the hardness of density estimation for $k$-mixtures of $n$-dimensional Gaussians by showing a reduction from homogeneous CLWE. This answers an open question regarding its computational complexity~\cite{diakonikolas2016structured,moitra2018}. 
We first formally define density estimation for Gaussian mixtures.

\begin{definition}[Density estimation of Gaussian mixtures]
Let $\cG_{n,k}$ be the family of $k$-mixtures of $n$-dimensional Gaussians. The problem of \emph{density estimation} for $\cG_{n,k}$ is the following. Given $\delta > 0$ and sample access to an unknown $P \in \cG_{n,k}$, with probability $9/10$, output a hypothesis distribution $Q$ (in the form of an evaluation oracle) such that $\Delta(Q,P) \le \delta$.
\end{definition}

For our purposes, we fix the precision parameter $\delta$ to a very small constant, say, $\delta = 10^{-3}$. Now we show a reduction from $\hclwe_{\beta,\gamma}$ to the problem of density estimation for Gaussian mixtures. Corollary~\ref{cor:hc} shows that $\hclwe_{\beta,\gamma}$ is hard for $\gamma \ge 2\sqrt{n}$ (assuming worst-case lattice problems are hard). Hence, by taking $\gamma = 2\sqrt{n}$ and $g(n) = O(\log n)$ in Proposition~\ref{prop:mixture-learning-hardness}, we rule out the possibility of a $\poly(n,k)$-time density estimation algorithm for $\cG_{n,k}$ under the same hardness assumption. 

\begin{proposition}
\label{prop:mixture-learning-hardness}
Let $\beta = \beta(n) \in (0,1/32)$, $\gamma = \gamma(n) \ge 1$, and $g(n) \ge 4\pi$. For $k = 2\gamma \sqrt{g(n)/\pi}$, if there is an $\exp(g(n))$-time algorithm that solves density estimation for $\cG_{n,2k+1}$, then there is a $O(\exp(g(n)))$-time algorithm that solves $\hclwe_{\beta,\gamma}$.
\end{proposition}
\begin{proof}
We apply the  density estimation algorithm $\cA$ to the unknown given distribution $P$. As we will show below, with constant probability, it outputs a density estimate $f$ that satisfies $\Delta(f,P) < 2\delta = 2 \cdot 10^{-3}$ (and this is even though $H_{\bw,\beta,\gamma}$ has infinitely many components). We then test whether $P = D_{\mathbb{R}^n}$ or not using the following procedure. We repeat the following procedure $m=1/(6\sqrt{\delta})$ times. We draw $\bx \sim D_{\mathbb{R}^n}$ and check whether the following holds
\begin{align}
    \frac{f(\bx)}{D(\bx)} \in [1-\sqrt{\delta},1+\sqrt{\delta}] \label{eqn:equality-test}\;,
\end{align}
where $D$ denotes the density of $D_{\mathbb{R}^n}$. We output $P = D_{\mathbb{R}^n}$ if Eq.~\eqref{eqn:equality-test} holds for all $m$ independent trials and $P = H_{\bw,\beta,\gamma}$ otherwise.
Since $\Delta(H_{\bw,\beta,\gamma},D_{\mathbb{R}^n}) > 1/2$ (Claim~\ref{claim:hclwe-tv-distance}), it is not hard to see that this test solves $\hclwe_{\beta,\gamma}$ with probability at least $2/3$ (see~\cite[Observation 24]{rubinfeld-servedio2009monotone} for a closely related statement). Moreover, the total running time is $O(\exp(g(n))$ since this test uses a constant number of samples.

If $P = D_{\mathbb{R}^n}$, it is obvious that $\cA$ outputs a close density estimate with constant probability since $D_{\mathbb{R}^n} \in \cG_{n,2k+1}$. It remains to consider the case $P = H_{\bw,\beta,\gamma}$. To this end, we observe that $H_{\bw,\beta,\gamma}$ is close to a $(2k+1)$-mixture of Gaussians. Indeed, by Claim~\ref{claim:hclwe-truncation} below, 
\begin{align}
    \Delta(H_{\bw,\beta,\gamma},H^{(k)}) \le 2\exp(-\pi\cdot k^2/(\beta^2+\gamma^2)) < 2\exp(-\pi \cdot k^2/(2\gamma^2)) \nonumber \;,
\end{align}
where $H^{(k)}$ is the distribution given by truncating $H_{\bw,\beta,\gamma}$ to the $(2k+1)$ central mixture components.
Hence, the statistical distance between the joint distribution of $\exp(g(n))$ samples from $H_{\bw,\beta,\gamma}$ and that of $\exp(g(n))$ samples from $H^{(k)}$ is bounded by
\begin{align}
2\exp(-\pi \cdot k^2/(2\gamma^2))\cdot\exp(g(n)) = 2\exp(-g(n)) \le 2\exp(-4\pi) \; .\nonumber
\end{align}
Since the two distributions are statistically close, a standard argument shows that $\cA$ will output $f$ satisfying $\Delta(f,H_{\bw,\beta,\gamma}) \le \Delta(f,H^{(k)}) + \Delta(H^{(k)},H_{\bw,\beta,\gamma}) < 2\delta$ with constant probability.
\end{proof}

\begin{claim}
\label{claim:hclwe-tv-distance}
Let $\beta = \beta(n) \in (0,1/32)$ and $\gamma = \gamma(n) \ge 1$. Then,
\begin{align*}
    \Delta(H_{\bw,\beta,\gamma},D_{\mathbb{R}^n}) > 1/2\;.
\end{align*}
\end{claim}
\begin{proof}
Let $\gamma' = \sqrt{\beta^2+\gamma^2} > \gamma$. Let $\by \in \mathbb{R}^n$ be a random vector distributed according to $H_{\bw,\beta,\gamma}$. Using the Gaussian mixture form of~\eqref{eqn:hclwe-mixture-def}, we observe that $\langle \by, \bw \rangle \bmod{\gamma/\gamma'^2}$ is distributed according to $D_{\beta/\gamma'} \bmod{\gamma/\gamma'^2}$. Since statistical distance cannot increase by applying a function (inner product with $\bw$ and then applying the modulo operation in this case), it suffices to lower bound the statistical distance between $D_{\beta/\gamma'} \bmod{\gamma/\gamma'^2}$ and $D \bmod{\gamma/\gamma'^2}$, where $D$ denotes the 1-dimensional standard Gaussian.

By Chernoff, for all $\zeta>0$, at least $1-\zeta$ mass of $D_{\beta/\gamma'}$ is contained in $[- a \cdot (\beta/\gamma'), a \cdot (\beta/\gamma')]$, where $a = \sqrt{\log(1/\zeta)}$. Hence, $D_{\beta/\gamma'} \bmod{\gamma/\gamma'^2}$ is at least $1-2a\beta \gamma'/\gamma-\zeta$ far in statistical distance from the uniform distribution over $\mathbb{R}/(\gamma/\gamma'^2)\mathbb{Z}$, which we denote by $U$. 
Moreover, by Lemma~\ref{lem:smoothing-uniform} and Lemma~\ref{lem:smoothing-dual}, $D \bmod{\gamma/\gamma'^2}$ is within statistical distance $\eps/2 = \exp(-\gamma'^4/\gamma^2)/2$ from $U$. Therefore,
\begin{align}
\Delta(D_{\beta/\gamma'} \bmod{\gamma/\gamma'^2},D \bmod{\gamma/\gamma'^2}) 
&\ge  \Delta(D_{\beta/\gamma'} \bmod{\gamma/\gamma'^2},U) - \Delta(U,D \bmod{\gamma/\gamma'^2}) \nonumber \\
&\ge 1-2a\beta\gamma'/\gamma-\zeta-\eps/2 \nonumber \\
&> 1-2\sqrt{2}a\beta-\zeta-\exp(-\gamma^2)/2 \label{eqn:hclwe-tv-plug-in-values}  \\
&> 1/2 \nonumber\;,
\end{align}
where we set $\zeta = \exp(-2)$ and use the fact that $\beta \le 1/32$ and $\gamma \ge 1$ in \eqref{eqn:hclwe-tv-plug-in-values}.
\end{proof}

\begin{claim}
\label{claim:hclwe-truncation}
Let $\beta = \beta(n) \in (0,1), \gamma = \gamma(n) \ge 1$, and $k \in \mathbb{Z}^{+}$. Then,
\begin{align}
    \Delta(H_{\bw,\beta,\gamma},H^{(k)}) \le 2\exp(-\pi\cdot k^2/(\beta^2+\gamma^2))\nonumber \;,
\end{align}
where $H^{(k)}$ is the distribution given by truncating $H_{\bw,\beta,\gamma}$ to the central $(2k+1)$ mixture components.
\end{claim}
\begin{proof}
We express $H_{\bw,\beta,\gamma}$ in its Gaussian mixture form given in Eq.~\eqref{eqn:hclwe-mixture-def} and define a random variable $X$ taking on values in $\mathbb{Z}$ such that the probability of $X = i$ is equal to the probability that a sample comes from the $i$-th component in $H_{\bw,\beta,\gamma}$. Then, we observe that $H^{(k)}$ is the distribution given by conditioning on $|X| \le k$. Since $X$ is a discrete Gaussian random variable with distribution $D_{\mathbb{Z},\sqrt{\beta^2+\gamma^2}}$, we observe that $\Pr[|X| > k] \le \eps := 2\exp(-\pi \cdot k^2/(\beta^2+\gamma^2))$ by~\cite[Lemma 2.8]{micciancio-peikert2012trapdoor}.
Since conditioning on an event of probability $1-\eps$ cannot change the statistical distance by more than $\eps$, we have
\begin{align}
    \Delta(H_{\bw,\beta,\gamma}, H^{(k)}) \le \eps \nonumber \;.
\end{align}
\end{proof}

\section{LLL Solves Noiseless CLWE}
\label{section:lll-clwe}
The noiseless CLWE problem ($\beta = 0$) can be solved in polynomial time using LLL. This applies both to the homogeneous and the inhomogeneous versions, as well as to the search version. The argument can be extended to the case of exponentially small $\beta>0$.

The key idea is to take samples $(\by_i, z_i)$, and find integer coefficients $c_1,\ldots,c_m$ such that $\by = \sum_{i=1}^m c_i \by_i$ is short, say 
$\|\by\| \ll 1/\gamma$. By Cauchy-Schwarz, we then have that $\gamma \langle \by, \bw \rangle = \sum_{i=1}^m c_i z_i$ over the reals (not modulo 1!). This is formalized in Theorem~\ref{thm:lll-noiseless-clwe}. We first state Minkowski's Convex Body Theorem, which we will use in the proof of our procedure.
\begin{lemma}[\cite{minkowski1910geometrie}]
\label{lem:minkowski-cvx}
Let $L$ be a full-rank $n$-dimensional lattice. Then, for any centrally-symmetric convex set $S$, if $\operatorname{vol}(S) > 2^n \cdot |\det(L)|$, then $S$ contains a non-zero lattice point.
\end{lemma}

\begin{theorem}
\label{thm:lll-noiseless-clwe}
Let $\gamma = \gamma(n)$ be a polynomial in $n$. Then, there exists a polynomial-time algorithm for solving $\clwe_{0,\gamma}$.
\end{theorem}
\begin{proof}
Take $n+1$ CLWE samples $\{(\by_i,z_i)\}_{i=1}^{n+1}$ and consider the matrix
\begin{align*}
Y = \begin{bmatrix}
    \by_1 & \cdots & \by_n & \by_{n+1} \\
    0 & \cdots & 0 & \delta 
    \end{bmatrix} \; ,
\end{align*}
where $\delta = 2^{-3n^2}$. 

Consider the lattice $L$ generated by the columns of $Y$. Since $\by_i$'s are drawn from the Gaussian distribution, $L$ is full rank. By Hadamard's inequality, and the fact that with probability exponentially close to $1$, $\|\by_i\| \leq  \sqrt{n}$ for all $i$, we have 
\begin{align*}
|\det(L)| \leq \delta \cdot n^{n/2} < 2^{-2n^2} \; .
\end{align*} 

Now consider the $n$-dimensional cube $S$ centered at $\bzero$ with side length $2^{-n}$. Then, $\operatorname{vol}(S) = 2^{-n^2}$, and by Lemma~\ref{lem:minkowski-cvx}, $L$ contains a vector $\bv$ satisfying $\|\bv\|_{\infty} \leq 2^{-n}$ and so $\| \bv \|_2 \leq \sqrt{n}\cdot 2^{-n}$. 
Applying the LLL algorithm~\cite{lenstra1982lll} gives us an integer combination of the columns of $Y$ whose length is within $2^{(n+1)/2}$ factor of the shortest vector in $L$, which will therefore have $\ell_2$ norm less than $\sqrt{n} \cdot 2^{-(n-1)/2}$. 
Let $\by$ be the corresponding combination of the $\by_i$ vectors (which is equivalently given by the first $n$ coordinates of the output of LLL) and 
$z \in (-1/2,1/2]$ a representative of the corresponding integer combination of the $z_i$ mod 1.
Then, we have $\|\by\|_2 \leq \sqrt{n} \cdot 2^{-(n-1)/2}$ and therefore we obtain the linear equation $\gamma \cdot \langle \by,\bw \rangle = z$ over the reals (without mod 1). 

We now repeat the above procedure $n$ times, and recover $\bw$ by solving the resulting $n$ linear equations. 
It remains to argue why the $n$ vectors $\by$ we collect are linearly independent. 
First, note that the output $\by$ is guaranteed to be a non-zero vector since with probability $1$, no integer combination of the Gaussian distributed $\by_i$ is $\bzero$.
Next, note that LLL is equivariant to rotations, i.e., if we rotate the input basis then the output vector will also be rotated by the same rotation. Moreover, spherical Gaussians are rotationally invariant. Hence, the distribution of the output vector $\by \in \mathbb{R}^n$
is also rotationally invariant. Therefore, repeating the above procedure $n$ times will give us $n$ linearly independent  vectors.
\end{proof}

\section{Subexponential Algorithm for Homogeneous CLWE}
\label{section:subexp}
For $\gamma = o(\sqrt{n})$, the covariance matrix will reveal the discrete structure of homogeneous CLWE, which will lead to a subexponential time algorithm for the problem. This clarifies why the hardness results of homogeneous CLWE do not extend beyond $\gamma \geq 2\sqrt{n}$.

We define \emph{noiseless homogeneous CLWE distribution} $H_{\bw, \gamma}$ as $H_{\bw, \beta, \gamma}$ with $\beta = 0$.
We begin with a claim that will allow us to focus on the noiseless case.

\begin{claim}
\label{claim:noiseless-is-sufficient}
By adding Gaussian noise $\ngauss{\beta/\gamma}$ to $H_{\bw,\gamma}$ and then rescaling by a factor of $\gamma/\sqrt{\beta^2+\gamma^2}$, the resulting distribution is $H_{\bw, \tilde{\beta}, \tilde{\gamma}}$, where $\tilde{\gamma} = \gamma/\sqrt{1+(\beta/\gamma)^2}$ and $\tilde{\beta} = \tilde{\gamma}(\beta/\gamma)$.\footnote{%
Equivalently, in terms of the Gaussian mixture representation of Eq.~\eqref{eqn:hclwe-mixture-def}, the resulting distribution has layers spaced by $1/\sqrt{\gamma^2+\beta^2}$ 
and of width $\beta/\sqrt{\gamma^2+\beta^2}$.
}
\end{claim}
\begin{proof}
Without loss of generality, suppose $\bw = \be_1$.

Let $\bz \sim H_{\bw,\gamma} + \ngauss{\beta/\gamma}$ and $\tilde{\bz} = \gamma\bz/\sqrt{\beta^2+\gamma^2}$.
It is easy to verify that the marginals density of $\tilde{\bz}$ on subspace $\be_1^\perp$ will simply be $\rho$.
Hence we focus on calculating the density of $z_1$ and $\tilde{z}_1$.
The density can be computed by convolving the probability densities of $H_{\bw,\gamma}$ and $\ngauss{\beta/\gamma}$ as follows. 
\begin{align*}
    H_{\bw,\gamma} * \ngauss{\beta/\gamma}(z_1) &\propto \sum_{k \in \mathbb{Z}}  \rho(k/\gamma)\cdot \rho_{\beta/\gamma}(z_1-k/\gamma) \\
    &= \rho_{\sqrt{\beta^2+\gamma^2}/\gamma}(z_1) \cdot \sum_{k \in \mathbb{Z}} \rho_{\beta/\sqrt{\beta^2+\gamma^2}}\Big(k / \gamma - \frac{\gamma^2}{\beta^2+\gamma^2}z_1 \Big) \\
    &= \rho(\tilde{z}_1) \cdot \sum_{k \in \mathbb{Z}} \rho_{\tilde{\beta}}\Big(k - \tilde{\gamma} \tilde{z}_1\Big)
    \; ,
\end{align*}
where the second to last equality follows from Claim~\ref{claim:complete-squares}.
This verifies that the resulting distribution is indeed $H_{\bw, \tilde{\beta}, \tilde{\gamma}}$.
\end{proof}

Claim~\ref{claim:noiseless-is-sufficient} implies an homogeneous CLWE distribution with $\beta > 0$ is equivalent to a noiseless homogeneous CLWE distribution with independent Gaussian noise added. We will first analyze the noiseless case and then derive the covariance of noisy (i.e., $\beta > 0$) case by adding independent Gaussian noise and rescaling.

\begin{lemma}
\label{lem:covariance-hclwe-noiseless}
Let $\Sigma \succ 0$ be the covariance matrix of the 
$n$-dimensional noiseless homogeneous CLWE distribution
$H_{\bw,\gamma}$ with $\gamma \ge 1$. Then,
\begin{align*}
    \Big\|\Sigma - \frac{1}{2\pi} I_n \Big\| \geq \gamma^2 \exp(-\pi\gamma^2) \; ,
\end{align*}
where $\|\cdot\|$ denotes the spectral norm.
\end{lemma}
\begin{proof}
Without loss of generality, let $\bw = \be_1$.
Then $H_{\bw,\gamma} = D_{L} \times D_{\mathbb{R}^{n-1}}$ where $L$ is the one-dimensional lattice $(1/\gamma)\mathbb{Z}$.
Then, $\Sigma = \operatorname{diag}(\E_{x \sim D_{L}}[x^2], \frac{1}{2\pi},\dots,\frac{1}{2\pi})$, so it suffices to show that
\begin{equation*}
    \Big| \E_{x \sim D_{L}}[x^2] - \frac{1}{2\pi} \Big| \ge \gamma^2 \exp(-\pi\gamma^2)
    \; .
\end{equation*}
Define $g(x) = x^2 \cdot \rho(x)$.
The Fourier transform of $\rho$ is itself; the Fourier transform of $g$ is given by
\begin{align*}
    \what{g}(y) = \Big(\frac{1}{2\pi}-y^2\Big) \rho(y)
    \; .
\end{align*}
By definition and Poisson's summation formula (Lemma~\ref{lem:poisson-sum}), we have
\begin{align}
    \E_{x \sim D_{L}}[x^2]
    &= \frac{g(L)}{\rho(L)} \nonumber \\
    &= \frac{\det(L^*)\cdot \what{g}(L^*)}{\det(L^*)\cdot \rho(L^*)}
    = \frac{\what{g}(L^*)}{\rho(L^*)} \nonumber \; ,
\end{align}
where $L^* = ((1/\gamma)\mathbb{Z})^* = \gamma \mathbb{Z}$.
Combining this with the expression for $\what{g}$, we have
\begin{align*}
    \Bigl|\E_{x \sim D_{L}}[x^2]-\frac{1}{2\pi}\Bigr| &= \frac{\sum_{y \in L^*}y^2\rho(y)}{1+\rho(L^*\setminus\{0\})} \\
    &\geq \gamma^2 \exp(-\pi \gamma^2) \; ,
\end{align*}
where we use the fact that for $\gamma \ge 1$,
\begin{align*}
    \rho(\gamma \mathbb{Z}\setminus\{0\}) \le 
    \rho(\mathbb{Z}\setminus\{0\}) &< 2\sum_{k=1}^\infty \exp(-\pi k) = \frac{2\exp(-\pi)}{1-\exp(-\pi)} < 1
    \; .
    \qedhere
\end{align*}
\end{proof}

Combining Claim~\ref{claim:noiseless-is-sufficient} and Lemma~\ref{lem:covariance-hclwe-noiseless}, we get the following corollary for the noisy case.

\begin{corollary}
\label{cor:covariance-hclwe}
Let $\Sigma \succ 0$ be the covariance matrix of
$n$-dimensional homogeneous CLWE distribution
$H_{\bw,\beta,\gamma}$ with $\gamma \ge 1$ and $\beta > 0$. Then,
\begin{align*}
    \Big\|\Sigma - \frac{1}{2\pi} I_n \Big\| \geq \gamma^2 \exp(-\pi(\beta^2+\gamma^2)) \; ,
\end{align*}
where $\|\cdot\|$ denotes the spectral norm.
\end{corollary}
\begin{proof}
Using Claim~\ref{claim:noiseless-is-sufficient}, we can view samples from $H_{\bw,\beta,\gamma}$ as samples from $H_{\bw,\gamma'}$ with independent Gaussian noise of width $\beta'/\gamma'$ added and rescaled by $\gamma'/\sqrt{\beta'^2+\gamma'^2}$, where $\beta', \gamma'$ are given by
\begin{align*}
    \beta' &= \beta \sqrt{1+(\beta/\gamma)^2} \; , \\
    \gamma' &= \sqrt{\beta^2+\gamma^2} \;.
\end{align*}
Let $\Sigma$ be the covariance of $H_{\bw,\beta,\gamma}$ and let $\Sigma_0$ be the covariance of $H_{\bw,\gamma'}$. Since the Gaussian noise added to $H_{\bw,\gamma'}$ is independent and $\beta'/\gamma' = \beta/\gamma$,
\begin{align*}
    \Sigma = \frac{1}{1+(\beta/\gamma)^2}\Big(\Sigma_0 + \frac{(\beta/\gamma)^2}{2\pi} I_n\Big) \;.
\end{align*}
Hence,
\begin{align*}
    \Big\|\Sigma - \frac{1}{2\pi}I_n\Big\| &= \frac{1}{1+(\beta/\gamma)^2} \Big\|\Big(\Sigma_0 + \frac{(\beta/\gamma)^2}{2\pi}I_n\Big)-\frac{1+(\beta/\gamma)^2}{2\pi}I_n\Big\| \\
    &= \frac{1}{1+(\beta/\gamma)^2}\Big\|\Sigma_0 - \frac{1}{2\pi}I_n \Big\| \\
    &\geq \gamma^2 \exp(-\pi(\beta^2+\gamma^2)) \; .
\end{align*}
where the last inequality follows from Lemma~\ref{lem:covariance-hclwe-noiseless}.
\end{proof}

We use the following lemma, which provides an upper bound on the error in estimating the covariance matrix by samples. The sub-gaussian norm of a random variable $Y$ is defined as $\|Y\|_{\psi_2} = \inf\{t > 0 \mid \mathbb{E}[\exp(Y^2/t^2)] \leq 2\}$ and that of an $n$-dimensional random vector $\by$ is defined as $\|\by\|_{\psi_2} = \sup_{\bu \in \mathbb{S}^{n-1}}\|\langle \by, \bu \rangle\|_{\psi_2}$.

\begin{lemma}[{\cite[Theorem 4.6.1]{vershynin2018high}}]
\label{lem:covariance-estimate}
Let $A$ be an $m\times n$ matrix whose rows $A_i$ are independent, mean zero, sub-gaussian isotropic random vectors in $\mathbb{R}^n$. Then for any $u \geq 0$ we have
\begin{align*}
    \Big\|\frac{1}{m}A^TA-I_n\Big\| \leq K^2 \max(\delta,\delta^2) \; \text{ where } \delta = C\Big(\sqrt{\frac{n}{m}} +\frac{u}{\sqrt{m}}\Big)\;,
\end{align*}
with probability at least $1-2e^{-u^2}$ for some constant $C > 0$. Here, $K = \max_i \|A_i\|_{\psi_i}$.
\end{lemma}

Combining Corollary~\ref{cor:covariance-hclwe} and Lemma~\ref{lem:covariance-estimate}, we have the following theorem for distinguishing homogeneous CLWE distribution and Gaussian distribution.

\begin{theorem}
\label{thm:subexp-hclwe}
Let $\gamma = n^{\eps}$, where $\eps < 1/2$ is a constant, and let $\beta = \beta(n) \in (0,1)$. Then, there exists an $\exp(O(n^{2\eps}))$-time algorithm that solves $\hclwe_{\beta,\gamma}$.
\end{theorem}

\begin{proof}
Our algorithm takes $m$ samples from the unknown input distribution $P$ and computes the sample covariance matrix $\Sigma_m = (1/m)A^TA$, where $A$'s rows are the samples, and its eigenvalues $\mu_1, \ldots, \mu_n$. Then, it determines whether $P$ is a homogeneous CLWE distribution or not by testing that
\begin{align*}
    \Bigl|\mu_i - \frac{1}{2\pi}\Bigr| \le \frac{1}{2}\cdot \gamma^2\exp(-\pi (\beta^2+\gamma^2)) \; \text{ for all } i \in [n]\;.
\end{align*}

The running time of this algorithm is $O(n^2 m) = \exp(O(n^{2\eps}))$. To show correctness, we first consider the case $P = D_{\mathbb{R}^n}$. The standard Gaussian distribution satisfies the conditions of Lemma~\ref{lem:covariance-estimate} (after rescaling by $1/(2\pi)$). Hence, the eigenvalues of $\Sigma_m$ will be within distance $O(\sqrt{n/m})$ from $1/(2\pi)$ with high probability.

Now consider the case $P = H_{\bw,\beta,\gamma}$. We can assume $\bw=\be_1$ without loss of generality since eigenvalues are invariant under rotations. Denote by $\by$ a random vector distributed according to $H_{\bw,\beta,\gamma}$ and $\sigma^2 = \E_{\by \sim H_{\bw,\beta,\gamma}}[y_1^2]$. The covariance of $P$ is given by
\begin{align}
    \Sigma = \begin{pmatrix} \sigma^2 & \bzero \\ \bzero & \frac{1}{2\pi}I_{n-1} \end{pmatrix} \label{eqn:hclwe-covariance-matrix} \; .
\end{align}
Now consider the sample covariance $\Sigma_m$ of $P$ and denote by $\sigma_m^2 = \bw^T\Sigma_m\bw = (1/m)\sum_{i=1}^m A_{i1}^2$. Since $A_{i1}$'s are sub-gaussian random variables~\cite[Lemma 2.8]{micciancio-peikert2012trapdoor}, $\sigma_m^2-\sigma^2$ is a sum of $m$ independent, mean-zero, sub-exponential random variables. For $m = \omega(n)$, Bernstein's inequality~\cite[Corollary 2.8.3]{vershynin2018high} implies that $|\sigma_m^2-\sigma^2| = O(\sqrt{n/m})$ with high probability. By Corollary~\ref{cor:covariance-hclwe}, we know that
\begin{align*}
\Big|\sigma^2 - \frac{1}{2\pi}\Big| \ge \gamma^2\exp(-\pi(\beta^2+\gamma^2)) \;.
\end{align*}

Hence, if we choose $m = \exp(c\gamma^2)$ with some sufficiently large constant $c$, then $\Sigma_m$ will have an eigenvalue that is noticeably far from $1/(2\pi)$ with high probability.
\end{proof}

\section{SQ Lower Bound for Homogeneous CLWE}
\label{section:sq-lb}
Statistical Query (SQ) algorithms~\cite{kearnsSQ1998} are a restricted class of algorithms that are only allowed to query expectations of functions of the input distribution without directly accessing individual samples. To be more precise, SQ algorithms access the input distribution indirectly via the STAT($\tau$) oracle, which given a query function $f$ and data distribution $D$, returns a value contained in the interval $\mathbb{E}_{x \sim D} [f(x)]+[-\tau, \tau]$ for some precision parameter $\tau$.

In this section, we prove SQ hardness of distinguishing homogeneous CLWE distributions from the standard Gaussian. In particular, we show that SQ algorithms that solve homogeneous CLWE require super-polynomial number of queries even with super-polynomial precision. This is formalized in Theorem~\ref{thm:hclwe-sq-lb}.

\begin{theorem}
\label{thm:hclwe-sq-lb}
Let $\beta = \beta(n) \in (0,1)$ and $\gamma = \gamma(n) \geq \sqrt{2}$. Then, any (randomized) SQ algorithm with precision $\tau \geq 4 \cdot \exp(-\pi \cdot \gamma^2/4)$ that successfully solves $\hclwe_{\beta, \gamma}$ with probability $\eta > 1/2$ requires at least $(2\eta-1)\cdot \exp(c n)\cdot \tau^2\beta^2/(4\gamma^2)$ statistical queries of precision $\tau$ for some constant $c > 0$.
\end{theorem}

Note that when $\gamma = \Omega(\sqrt{n})$ and $\gamma/\beta = \poly(n)$, even
exponential precision $\tau = \exp(-O(n))$ results in a query lower bound that grows as $\exp(\tilde{\Omega}(n))$. This establishes an unconditional hardness result for SQ algorithms in the parameter regime $\gamma = \Omega(\sqrt{n})$, which is consistent with our computational hardness result based on worst-case lattice problems. The uniform spacing in homogeneous CLWE distributions gives us tight control over their pairwise correlation (see definition in \eqref{eqn:pairwise-corr}), which leads to a simple proof of the SQ lower bound.

We first provide some necessary background on the SQ framework. We denote by $\cB(\cU,D)$ the decision problem in which the input distribution $P$ either equals $D$ or belongs to $\cU$, and the goal of the algorithm is to identify whether $P=D$ or $P \in \cU$. For our purposes, $D$ will be the standard Gaussian $D_{\mathbb{R}^n}$ and $\cU$ will be a finite set of homogeneous CLWE distributions. Abusing notation, we denote by $D(x)$ the density of $D$. Following \cite{feldman2017planted-clique}, we define the \emph{pairwise correlation} between two distributions $P, Q$ relative to $D$ as
\begin{align}
    \chi_D(P,Q) := \mathbb{E}_{\bx \sim D} \left[\left(\frac{P(\bx)}{D(\bx)}-1 \right)\cdot\left(\frac{Q(\bx)}{D(\bx)}-1 \right) \right] = \mathbb{E}_{\bx \sim D} \left[\frac{P(\bx)Q(\bx)}{D(\bx)^2}\right] -1 \label{eqn:pairwise-corr}\; .
\end{align}

Lemma~\ref{lem:decision-lb} below establishes a lower bound on the number of statistical queries required to solve $\cB(\cU,D)$ in terms of pairwise correlation between distributions in $\cU$.

\begin{lemma}[{\cite[Lemma 3.10]{feldman2017planted-clique}}]
\label{lem:decision-lb}
Let $D$ be a distribution and $\cU$ be a set of distributions both over a domain $X$ such that for any $P, Q \in \cU$
\begin{align}
    |\chi_D(P,Q)| \leq \begin{cases} \delta &\mbox{if } P = Q \\ 
    \eps &\mbox{otherwise }\;  \end{cases} \nonumber\;.
\end{align}
Let $\tau \ge \sqrt{2\eps}$. Then, any (randomized) SQ algorithm that solves $\cB(\cU,D)$ with success probability $\eta > 1/2$ requires at least $(2\eta-1)\cdot|\cU|\cdot\tau^2/(2\delta)$ queries to $\operatorname{STAT}(\tau)$.
\end{lemma}

The following proposition establishes a tight upper bound on the pairwise correlation between homogeneous CLWE distributions. To deduce Theorem~\ref{thm:hclwe-sq-lb} from Lemma~\ref{lem:decision-lb} and Proposition~\ref{prop:avg-corr}, we take a set of unit vectors $\cU$ such that any two distinct vectors $\bv, \bw \in \cU$ satisfy $|\langle \bv, \bw \rangle| \leq 1/\sqrt{2}$, and identify it with the set of homogeneous CLWE distributions $\{H_{\bw,\beta,\gamma}\}_{\bw \in \cU}$. A standard probabilistic argument shows that such a $\cU$  can be as large as $\exp(\Omega(n))$, which proves Theorem~\ref{thm:hclwe-sq-lb}.

\begin{proposition}
\label{prop:avg-corr}
Let $\bv, \bw \in \mathbb{R}^n$ be unit vectors and let $H_{\bv}, H_{\bw}$ be $n$-dimensional homogeneous CLWE distributions with parameters $\gamma \geq 1, \beta \in (0,1)$, and hidden direction $\bv$ and $\bw$, respectively. Then, for any $\alpha \ge 0$ that satisfies $\gamma^2(1-\alpha^2) \ge 1$,
\begin{align}
    |\chi_{D}(H_{\bv},H_{\bw})| \leq \begin{cases} 2(\gamma/\beta)^2 &\text{ if } \bv = \bw \\ 8\exp(-\pi\cdot \gamma^2(1-\alpha^2)) &\text{ if } |\langle \bv, \bw \rangle| \leq \alpha \end{cases} \nonumber \;.
\end{align}
\end{proposition}
\begin{proof}
We will show that computing $\chi_D(H_{\bv},H_{\bw})$ reduces to evaluating the Gaussian mass of two lattices $L_1$ and $L_2$ defined below. Then, we will tightly bound the Gaussian mass using Lemma~\ref{lem:poisson-sum} and Lemma~\ref{lem:smoothing-primal}, which will result in upper bounds on $|\chi_D(H_{\bv},H_{\bw})|$. We define $L_1$ and $L_2$ by specifying their bases $B_1$ and $B_2$, respectively.

\begin{align*}
    B_1 &= \frac{1}{\sqrt{\beta^2+\gamma^2}} \begin{pmatrix} 1 & 0 \\
    0 & 1 \end{pmatrix}
    \; ,\\
    B_2 &= \frac{1}{\sqrt{\beta^2+\gamma^2}}\begin{pmatrix} 1 & 0 \\
    -\frac{\alpha \gamma^2}{\zeta\sqrt{\beta^2+\gamma^2}} & \frac{\sqrt{\beta^2+\gamma^2}}{\zeta} \end{pmatrix}
    \; ,
\end{align*}
where $\zeta = \sqrt{(\beta^2+\gamma^2) -\alpha^2\gamma^4/(\beta^2+\gamma^2)}$. Then the basis of the dual lattice $L_1^*$ and $L_2^*$ is $B_1^{-T}$ and $B_2^{-T}$, respectively. Note that $\lambda_2(L_1)^2 = 1/(\beta^2+\gamma^2)$ and that the two columns of $B_2$ have the same norm, and so
\begin{align}
    \lambda_2(L_2)^2 &\leq \frac{1}{\beta^2+\gamma^2}\cdot \max\Big\{1+\frac{\alpha^2\gamma^4}{\zeta^2(\beta^2+\gamma^2)},\frac{\beta^2+\gamma^2}{\zeta^2}\Big\} \nonumber\\
    &= \frac{1}{\zeta^2} \label{eqn:lambda2-general} \\
    &\leq \frac{1}{\gamma^2(1-\alpha^2)} \label{eqn:lambda2-simple} \; .
\end{align}

Now define the density ratio $a(t) := H(t)/D(t)$, where $D$ is the standard Gaussian and $H$ is the marginal distribution of homogeneous CLWE with parameters $\beta, \gamma$ along the hidden direction. We immediately obtain
\begin{align}
    a(t) &= \frac{1}{Z} \sum_{k \in \mathbb{Z}} \rho_{\beta/\gamma}(t-k/\gamma) \label{eq:sq-density-ratio}
    \; ,
\end{align}
where $Z = \int_\mathbb{R} \rho(t) \cdot \sum_{k \in \mathbb{Z}} \rho_{\beta/\gamma}(t-k/\gamma) dt$. By Eq.~\eqref{eqn:hclwe-def-normalization}, $Z$ is given by
\begin{align*}
    Z = \frac{\beta}{\sqrt{\beta^2+\gamma^2}} \cdot \rho\Bigg(\frac{1}{\sqrt{\beta^2+\gamma^2}}\mathbb{Z}\Bigg) \; .
\end{align*}

Moreover, we can express $Z^2$ in terms of the Gaussian mass of $(L_1)$ as
\begin{align*}
    Z^2 = \frac{\beta^2}{\beta^2+\gamma^2}\cdot \rho(L_1) \; .
\end{align*}

$\chi_D(H_{\bv},H_{\bw})$ can be expressed in terms of $a(t)$ as
\begin{align}
    \chi_D(H_{\bv},H_{\bw}) = \E_{\bx \sim D}\Big[a(\langle \bx, \bw \rangle)\cdot a(\langle \bx, \bv \rangle)\Big] - 1 \label{eqn:avg-corr-simplified}\;.
\end{align}

Without loss of generality, assume $\bv = \be_1$ and $\bw = \alpha \be_1 + \xi \be_2$, where $\xi = \sqrt{1-\alpha^2}$. We first compute the pairwise correlation for $\bv \neq \bw$. For notational convenience, we denote by $\eps = 8\cdot\exp(-\pi\cdot\gamma^2(1-\alpha^2))$.

\begin{align}
    \chi_{D}(H_{\bv},H_{\bw}) + 1&= \E_{\bx \sim D} \Big[a(x_1)\cdot a(\alpha x_1 + \xi x_2)\Big] \nonumber \\
    &= \frac{1}{Z^2}\sum_{k, \ell \in \mathbb{Z}} \int \int \rho_{\beta}(\gamma x_1 - k)\cdot \rho_{\beta}((\gamma \alpha x_1 + \gamma \xi x_2) - \ell)\cdot \rho(x_1) \cdot \rho(x_2) dx_1 dx_2 \nonumber \\
    &= \frac{1}{Z^2}\cdot\frac{\beta}{\sqrt{(\gamma\xi)^2+\beta^2}}\sum_{k, \ell \in \mathbb{Z}} \int \rho_{\beta}(\gamma x_1 - k) \cdot \rho(x_1) \cdot \rho_{\sqrt{1+\beta^2/(\gamma\xi)^2}} (\ell/(\gamma\xi)-(\alpha/\xi) x_1) dx_1 \nonumber\\
    &= \frac{1}{Z^2}\cdot\frac{\beta}{\sqrt{(\gamma\xi)^2+\beta^2}}\cdot\frac{\beta\sqrt{(\gamma\xi)^2+\beta^2}}{\zeta\sqrt{\beta^2+\gamma^2}} \sum_{k, \ell \in \mathbb{Z}} \rho_{\sqrt{\beta^2+\gamma^2}}(k) \cdot \rho_{\zeta}\Big(\ell - \gamma^2 \alpha \cdot k/(\beta^2+\gamma^2)\Big) \nonumber\\
    &= \frac{\sqrt{\beta^2+\gamma^2}}{\zeta}\cdot \frac{\sum_{k, \ell \in \mathbb{Z}} \rho_{\sqrt{\beta^2+\gamma^2}}(k) \cdot \rho_{\zeta}\Big(\ell - \gamma^2 \alpha \cdot k/(\beta^2+\gamma^2)\Big)}{\rho(L_1)} \nonumber\\
    &= \frac{\sqrt{\beta^2+\gamma^2}}{\zeta} \cdot \frac{\rho(L_2)}{\rho(L_1)} \nonumber\\
    &= \frac{\sqrt{\beta^2+\gamma^2}}{\zeta}\cdot\frac{\det(L_2^*)}{\det(L_1^*)} \cdot \frac{\rho(L_2^*)}{\rho(L_1^*)}\nonumber \\
    &= \frac{\rho(L_2^*)}{\rho(L_1^*)} \label{eqn:poisson-distinct-directions}\\
    &\in \Big[\frac{1}{1+\eps}, 1+\eps\Big] \nonumber \; ,
\end{align}
In \eqref{eqn:poisson-distinct-directions}, we used the Poisson summation formula (Lemma~\ref{lem:poisson-sum}). The last line follows from \eqref{eqn:lambda2-simple} and Lemma~\ref{lem:smoothing-primal}, which implies that for any 2-dimensional lattice $L$ satisfying $\lambda_2(L) \leq 1$,
\begin{align}
    \rho(L^*\setminus\{\bzero\}) \leq 8\exp(-\pi/\lambda_2(L)^2) \label{eqn:2d-gaussian-mass-bound}\; .
\end{align}

Now consider the case $\bv = \bw$. Using \eqref{eqn:lambda2-general}, we get an upper bound $\lambda_2(L_2) \leq 1/\beta$ when $\alpha = 1$. It follows that $\lambda_2((\beta/\gamma)L_2) \le 1/\gamma \le 1$. Hence,
\begin{align}
    \chi_{D}(H_{\bv},H_{\bv}) + 1
    &= \frac{\sqrt{\beta^2+\gamma^2}}{\zeta} \cdot \frac{\rho(L_2)}{\rho(L_1)} \nonumber\\
    &\leq \frac{\sqrt{\beta^2+\gamma^2}}{\zeta}\cdot \frac{\rho((\beta/\gamma)L_2)}{\rho(L_1)} \nonumber \\
    &= \frac{\sqrt{\beta^2+\gamma^2}}{\zeta} \cdot \frac{\det((\gamma/\beta)L_2^*)}{\det(L_1^*)} \cdot \frac{\rho((\gamma/\beta)L_2^*)}{\rho(L_1^*)} \nonumber \\
    &= \frac{\gamma^2}{\beta^2}\cdot\frac{\rho((\gamma/\beta)L_2^*)}{\rho(L_1^*)} \label{eqn:poisson-same-direction}\\
    &\leq 2(\gamma/\beta)^2 \label{eqn:chi-correlation-ub} \; .
\end{align}
where we used Lemma~\ref{lem:poisson-sum} in \eqref{eqn:poisson-same-direction} and in \eqref{eqn:chi-correlation-ub}, we used \eqref{eqn:2d-gaussian-mass-bound} and the fact that $\lambda_2((\beta/\gamma)L_2) \leq 1$ to deduce $\rho((\gamma/\beta)L_2^*\setminus\{\bzero\}) \leq 1$.
\end{proof}

\section{Extension of Homogeneous CLWE to \texorpdfstring{$m \ge 1$}{m>=1} Hidden Directions}
\label{section:k-hc}

In this section, we generalize the hardness result to the setting where the homogeneous CLWE distribution has $m \ge 1$ hidden directions.
The proof is a relatively standard hybrid argument.

\begin{definition}[$m$-Homogeneous CLWE distribution]
For $0 \le m \le n$, matrix $\bW \in \mathbb{R}^{n \times m}$ with orthonormal columns $\bw_1,\ldots,\bw_m$, and $\beta, \gamma > 0$, define the \emph{$m$-homogeneous CLWE distribution} $H_{\bW, \beta, \gamma}$ over $\mathbb{R}^n$ to have density at $\by$ proportional to

\begin{align*}
    \rho(\by) \cdot \prod_{i = 1}^m \sum_{k \in \mathbb{Z}} \rho_\beta(k-\gamma\langle \by, \bw_i \rangle)
    \; .
\end{align*}
\end{definition}

Note that the $0$-homogeneous CLWE distribution is just $D_{\mathbb{R}^n}$ regardless of $\beta$ and $\gamma$.

\begin{definition} For parameters $\beta, \gamma > 0$ and $1 \le m \le n$, the average-case decision problem $\hclwe_{\beta, \gamma}^{(m)}$ is to distinguish the following two distributions over $\mathbb{R}^n$: (1) the $m$-homogeneous CLWE distribution $H_{\bW, \beta, \gamma}$ for some matrix $\bW \in \mathbb{R}^{n \times m}$ (which is fixed for all samples) with orthonormal columns chosen uniformly from the set of all such matrices, or (2) $D_{\mathbb{R}^n}$.
\end{definition}

\begin{lemma}
\label{lem:hc-to-k-hc}
For any $\beta, \gamma > 0$ and positive integer $m = m(n)$ such that $m \le n$ and $n - m = \Omega(n^c)$ for some constant $c > 0$,
if there exists an efficient algorithm that solves $\hclwe_{\beta,\gamma}^{(m)}$ with non-negligible advantage,
then there exists an efficient algorithm that solves $\hclwe_{\beta,\gamma}$ with non-negligible advantage.
\end{lemma}

\begin{proof}
Suppose $\cA$ is an efficient algorithm that solves $\hclwe_{\beta,\gamma}^{(m)}$ with non-negligible advantage
in dimension $n$.
Then consider the following algorithm $\cB$ that uses $\cA$ as an oracle and solves $\hclwe_{\beta,\gamma}$ in dimension $n' = n-m+1$.
\begin{enumerate}
    \item Input: $n'$-dimensional samples, drawn from either $\hclwe_{\beta,\gamma}$ or $D_{\mathbb{R}^{n'}}$;
    \item Choose $0 \le i \le m-1$ uniformly at random;
    \item Append $m-1 = n-n'$ coordinates to the given samples, where the first $i$ appended coordinates are drawn from $H_{\bI_i, \beta, \gamma}$ (with $\bI_i$ denoting the rank-$i$ identity matrix) and the rest of the coordinates are drawn from $D_{\mathbb{R}^{m - i -1}}$;
    \item Rotate the augmented samples using a uniformly random rotation from the orthogonal group $O(n)$;
    \item Call $\cA$ with the samples and output the result.
\end{enumerate}
As $n = O({n'}^{1/c})$, $\cB$ is an efficient algorithm.
Moreover, the samples passed to $\cA$ are effectively drawn from either $\hclwe_{\beta,\gamma}^{(i+1)}$ or $\hclwe_{\beta,\gamma}^{(i)}$.
Therefore the advantage of $\cB$ is at least $1/m$ fraction of the advantage of $\cA$, which would be non-negligible (in terms of $n$, and thus also in terms of $n'$) as well.
\end{proof}

Combining Corollary~\ref{cor:hc} and Lemma~\ref{lem:hc-to-k-hc}, we obtain the following corollary.

\begin{corollary}
For any $\beta = \beta(n) \in (0,1)$ and $\gamma = \gamma(n) \geq 2\sqrt{n}$ such that $\gamma/\beta$ is polynomially bounded,
and positive integer $m = m(n)$ such that $m \le n$ and $n - m = \Omega(n^c)$ for some constant $c > 0$,
there is a polynomial-time quantum reduction from $\dgs_{2\sqrt{2 n}\eta_\eps(L)/\beta}$ to $\hclwe_{\beta,\gamma}^{(m)}$.
\end{corollary}

\printbibliography
 
\end{document}

%% file: plots/clwe-iteration.tex
\begin{tikzpicture}

\node (begin) at (3.5, 6) {\huge{\bf \dots}};
\node[draw, fill=white, drop shadow={color=black}, align=center, text width=7em] (DGS1) at (0, 5) {poly samples from $D_{L, r}$};
\node[draw, fill=white, drop shadow={color=black}, align=center, text width=7em] (BDD1) at (7, 4) {oracle for $\mathrm{BDD}_{L^*,\gamma / (\sqrt2 r)}$};
\node[draw, fill=white, drop shadow={color=black}, align=center, text width=7em] (DGS2) at (0, 3) {poly samples from $D_{L, r \sqrt{n} / \gamma}$};
\node[draw, fill=white, drop shadow={color=black}, align=center, text width=7em] (BDD2) at (7, 2) {oracle for $\mathrm{BDD}_{L^*,\gamma^2 / (\sqrt{2n} r)}$};
\node[draw, fill=white, drop shadow={color=black}, align=center, text width=7em] (DGS3) at (0, 1) {poly samples from $D_{L, r n / \gamma^2}$};
\node (end) at (3.5, 0) {\huge{\bf \dots}};

\draw[-latex, thick] (begin) -- (DGS1);
\draw[-latex, thick] (DGS1) -- node[above, sloped] {\small{classical, uses CLWE}} ++ (BDD1);
\draw[-latex, thick] (BDD1) -- node[above, sloped] {\small{quantum}} ++ (DGS2);
\draw[-latex, thick] (DGS2) -- node[above, sloped] {\small{classical, uses CLWE}} ++ (BDD2);
\draw[-latex, thick] (BDD2) -- node[above, sloped] {\small{quantum}} ++ (DGS3);
\draw[-latex, thick] (DGS3) -- (end);

\end{tikzpicture}